\documentclass[copyright,creativecommons]{eptcs}
\hypersetup{hidelinks = true}
\usepackage{enumitem,latexsym,breakurl,amsmath,amssymb,color,colortbl,makeidx}
\setlist{nolistsep}  
\DeclareSymbolFont{cmi}   {OT1}{cmr}{m}{it}
\DeclareMathSymbol\val    {\mathord}{cmi}{118} 
\DeclareMathSymbol\wal    {\mathord}{cmi}{119} 
\newfont{\bbb}{bbm10 scaled 1100}                     
\newfont{\bbbs}{bbm10 scaled 800}                     
\newfont{\bbbss}{bbm10 scaled 600}                    
\newcommand{\IN}{\mbox{\bbb N}}                       
\newcommand{\ID}{\mbox{\bbb D}}                       
\newcommand{\IDs}{\mbox{\bbbs D}}                     
\newcommand{\IG}{\mbox{\bbb G}}                       
\newcommand{\IGs}{\mbox{\bbbs G}}                     
\newcommand{\IGss}{\mbox{\bbbss G}}                   
\newcommand{\IT}{\mbox{\bbb T}}                       
\newcommand{\T}{{\rm T}}                              
\newcommand{\fL}{{\cal L}}                            
\newcommand{\fT}{{\cal T}}                            
\newcommand{\R}{\mathrel{\cal R}}                     
\newtheorem{defi}{Definition}
\newtheorem{theo}{Theorem}
\newtheorem{prop}{Proposition}
\newtheorem{obse}{Observation}
\newtheorem{exam}{Example}

\newenvironment{definition}[1]{\begin{defi} \rm \label{df:#1} }{\end{defi}}
\newenvironment{theorem}[1]{\begin{theo} \rm \label{thm:#1} }{\end{theo}}
\newenvironment{proposition}[1]{\begin{prop} \rm \label{pr:#1} }{\end{prop}}
\newenvironment{observation}[1]{\begin{obse} \rm \label{obs:#1} }{\end{obse}}
\newenvironment{example}[1]{\begin{exam} \rm \label{ex:#1} }{\end{exam}}
\newenvironment{proof}{\begin{trivlist} \item[\hspace{\labelsep}\bf Proof:]}{\hfill $\Box$\end{trivlist}}

\newcommand{\Sec}[1]{Section~\ref{sec:#1}}
\newcommand{\df}[1]{Definition~\ref{df:#1}}
\newcommand{\thm}[1]{Theorem~\ref{thm:#1}}
\newcommand{\pr}[1]{Proposition~\ref{pr:#1}}

\newcommand{\ex}[1]{Example~\ref{ex:#1}}
\makeatletter
\def\comesfrom{\@transition\leftarrowfill}
\def\goesto{\@transition\rightarrowfill}
\def\ngoesto{\@transition\nrightarrowfill}
\def\Goesto{\@transition\Rightarrowfill}
\def\nGoesto{\@transition\nRightarrowfill}
\def\xmapsto{\@transition\mapstofill}
\def\nxmapsto{\@transition\nmapstofill}
\def\@transition#1{\@@transition{#1}}
\newbox\@transbox
\newbox\@arrowbox
\newbox\@downbox
\def\@@transition#1#2%
   {\setbox\@transbox\hbox
      {\vrule height 1.5ex depth .8ex width 0ex\hskip0.25em$\scriptstyle#2$\hskip0.25em}
   \ifdim\wd\@transbox<1.5em
      \setbox\@transbox\hbox to 1.5em{\hfil\box\@transbox\hfil}\fi
   \setbox\@arrowbox\hbox to \wd\@transbox{#1}
   \ht\@arrowbox\z@\dp\@arrowbox\z@
   \setbox\@transbox\hbox{$\mathop{\box\@arrowbox}\limits^{\box\@transbox}$}
   \dp\@transbox\z@\ht\@transbox 10pt
   \mathrel{\box\@transbox}}
\def\nrightarrowfill{$\m@th\mathord-\mkern-6mu%
  \cleaders\hbox{$\mkern-2mu\mathord-\mkern-2mu$}\hfill
  \mkern-6mu\mathord\not\mkern-2mu\mathord\rightarrow$}
\def\Rightarrowfill{$\m@th\mathord=\mkern-6mu%
  \cleaders\hbox{$\mkern-2mu\mathord=\mkern-2mu$}\hfill
  \mkern-6mu\mathord\Rightarrow$}
\def\nRightarrowfill{$\m@th\mathord=\mkern-6mu%
  \cleaders\hbox{$\mkern-2mu\mathord=\mkern-2mu$}\hfill
  \mkern-6mu\mathord\not\mathord\Rightarrow$}
\def\mapstofill{$\m@th\mathord\mapstochar\mathord-\mkern-6mu%
  \cleaders\hbox{$\mkern-2mu\mathord-\mkern-2mu$}\hfill
  \mkern-6mu\mathord\rightarrow$}
\def\nmapstofill{$\m@th\mathord\mapstochar\mathord-\mkern-6mu%
  \cleaders\hbox{$\mkern-2mu\mathord-\mkern-2mu$}\hfill
  \mkern-6mu\mathord\not\mkern-2mu\mathord\rightarrow$}
\makeatother 
\newcommand{\ar}[1]{\mathrel{\goesto{#1}}}            
\newcommand{\nar}[1]{\mathrel{\ngoesto{#1\;}}}        
\newcommand{\phrase}[1]{\index{#1}\emph{#1}}          
\newcommand{\plat}[1]{\raisebox{0pt}[0pt][0pt]{$#1$}} 
\newcommand{\den}[1]{\mbox{$[\hspace{-1.6pt}[$}\,#1\,\mbox{$]\hspace{-1.6pt}]$}}  
\newcommand{\rec}[1]{\plat{			      
	\stackrel{\mbox{\tiny $/$}}
	{\raisebox{-.3ex}[.3ex]{\tiny $\backslash$}}
	\!\!#1\!\!
	\stackrel{\mbox{\tiny $\backslash$}}
	{\raisebox{-.3ex}[.3ex]{\tiny $/$}} }}
\newcommand{\obis}[2]{\mathrel{_{#1}\,		      
	\raisebox{.3ex}{$\underline{\makebox[.7em]{$\leftrightarrow$}}$}
                  \,_{#2}}}
\newcommand{\bis}[1]{\obis{}{#1}}		      
\newcommand{\nobis}{\mathrel{\mbox{$\hspace{3.3pt}\not\hspace{-3.3pt}\,  
	\raisebox{.3ex}{$\underline{\makebox[.7em]{$\leftrightarrow$}}$}$}}}
\newcommand{\dom}{{\it dom}}                          
\newcommand{\eqa}{\mathrel{\plat{\stackrel{\alpha}=}}} 
\newcommand{\Var}{{\it Var}}                          
\newcommand{\var}{{\it var}}                          
\def\authorrunning{R.J. van Glabbeek}
\def\titlerunning{On the Meaning of Transition System Specifications}
\title{\titlerunning}
\author{Rob van Glabbeek
\institute{Data61, CSIRO, Sydney, Australia}
\institute{School of Computer Science and Engineering,
University of New South Wales, Sydney, Australia}
\email{rvg@cs.stanford.edu}}

\begin{document}

\maketitle

\begin{abstract}
Transition System Specifications provide programming and specification languages with a semantics. They
provide the meaning of a closed term as a \emph{process graph}: a state in a labelled transition system.
At the same time they provide the meaning of an $n$-ary operator, or more generally
an open term with $n$ free variables, as an $n$-ary operation on process graphs.
The classical way of doing this, the \emph{closed-term semantics}, reduces the
meaning of an open term to the meaning of its closed instantiations. It makes the meaning of an
operator dependent on the context in which it is employed. Here I propose an alternative
\emph{process graph semantics} of TSSs that does not suffer from this drawback.

Semantic equivalences on process graphs can be lifted to open terms conform either the closed-term
or the process graph semantics. For pure TSSs the latter is more discriminating.

I consider five sanity requirements on the semantics of programming and specification
languages equipped with a recursion construct: \emph{compositionality}, applied to $n$-ary operators,
recursion and variables, \emph{invariance under $\alpha$-conversion}, and the
\emph{recursive definition principle}, saying that the meaning of a recursive call should be a
solution of the corresponding recursion equations.
I establish that the satisfaction of four of these requirements under the
closed-term semantics of a TSS implies their satisfaction under the process graph semantics.
\end{abstract}

\section{Introduction}

Transition System Specifications (TSSs) \cite{GrV92} are a formalisation of \emph{Structural Operational
Semantics} \cite{Pl04} providing programming and specification languages with an interpretation. They
provide the meaning of a closed term as a \emph{process graph}: a state in a labelled transition system.
At the same time they provide the meaning of an $n$-ary operator of the language, or more generally
an open term with $n$ free variables, as an $n$-ary operation on process graphs.
The classical way of doing this proceeds by reducing the meaning of an open term to the meaning of
its closed instantiations. I call this the \emph{closed-term semantics} of TSSs.
A serious shortcoming of this approach is that it makes the meaning of an operator dependent on the
context in which it is employed.

\begin{example}{context dependence}
Consider a TSS featuring unary operators $f$, $id$ and $a.\_$ for each action
$a$ drawn from an alphabet $A$, and a constant $0$. The set of admitted transition labels is
$Act:=A\uplus\{\tau\}$. The transition rules are\vspace{-2ex}
\[\quad a.x \ar{a} x~~\mbox{\small (for all $a\in A$)} \qquad
  \frac{x \ar{a} x'}{f(x) \ar{a} f(x')}~~\mbox{\small (for all $a\in A$)} \qquad
  \frac{x \ar{\alpha} x'}{id(x) \ar{\alpha} id(x')}~~\mbox{\small (for all $\alpha\in Act$)} \qquad
\]
When considered in their own right, the operators $f$ and $id$ are rather different:
the latter can mimic $\tau$-transitions of its argument, and the former can not.
Yet, in the context of the given TSS, one has $f(p)\bis{}id(p)$, no matter which term $p$ is
substituted for the argument of these operators. Here $\bis{}$ denotes strong bisimulation
equivalence, as defined in \cite{Mi90ccs,vG00}. This is because no process in the given TSS ever
generates a transition with the label $\tau$. The identification of $f$ and $id$ up to $\bis{}$
can be considered unfortunate, one reason being that is ceases to hold as soon as the language is
enriched with a fresh operator $\tau.\_$ with the transition rule $\tau.x \ar{\tau} x$.
As I will show later, this invalidates an intuitively plausible theorem on the relative expressiveness
of specification languages.
\end{example}
Here I propose an alternative \emph{process graph semantics} of TSSs that does not suffer from this drawback.

In \cite{vG94} I proposed five requirements on the semantics of programming and specification
languages equipped with a recursion construct: \emph{compositionality}, applied to $n$-ary operators,
recursion and variables, \emph{invariance under $\alpha$-conversion}, and the
\emph{recursive definition principle}, saying that the meaning of a recursive call should be a
solution of the corresponding recursion equations.

In many prior works on structural operational semantics, (some of) these requirements have been shown to hold
for various TSSs when employing the closed-term semantics. It would be time consuming to redo all
that work for the process graph semantics proposed here. To prevent this I show that the
satisfaction of four of these requirements under the closed-term semantics of a TSS implies their
satisfaction under the process graph semantics. The remaining requirement holds almost always.
\vspace{-1ex}

\paragraph{Overview of the paper}

\Sec{syntax} presents the syntax of the programming and specification languages I consider here.
For simplicity I restrict myself to languages with single-sorted signature, optionally featuring a recursion construct.
This is a rich enough setting to include process algebras like CCS~\cite{Mi90ccs}, CSP~\cite{BHR84},
ACP~\cite{BK84}, {\sc Meije}~\cite{AB84,dS85} and SCCS~\cite{Mi83}.

The traditional ``\emph{closed-term}'' interpretation of the process calculi CCS, {\sc Meije} and SCCS effectively
collapses syntax and semantics by interpreting the entire language as one big labelled transition
system (LTS) in which the closed terms of the language constitute the set of states. This LTS is
generated by a TSS, as formally defined in \Sec{TSS}. Semantic equivalences on LTSs thereby directly
relate closed terms. Two open terms are judged equivalent iff each of their closed substitutions are.

In \Sec{semantics} I present an interpretation of programming and specification languages that is
more common in universal algebra and mathematical logic \cite{Ma77}, and is also used in the
traditional semantics of ACP and CSP\@. It separates syntax and semantics through a semantic mapping
that associates with each closed term a value, and with each open term an operation on values.
This matches with what often is called \emph{denotational semantics}, except that I do not require
the meaning of recursion constructs to be provided by means of fixed point techniques.
In \Sec{process graphs} I specialise this general approach to an operational one by taking the
values to be \emph{process graphs}: states in labelled transition systems.
Likewise, \Sec{closed-term} casts the closed-term interpretation as a special case of the approach
from \Sec{semantics}, by taking the values to be the closed terms.

\Sec{semantics} also formulates the five sanity requirements mentioned above, most in a couple of equivalent forms.
\Sec{closed-term} shows how these requirements simplify to better recognisable forms under the
closed-term interpretation of programming and specification languages.
These requirements are parametrised by the choice of a semantic equivalence $\sim$ on values,
relating values that one does not need to distinguish. The traditional treatments of universal algebra
and mathematical logic, and the process algebra CSP, do not involve such a semantic equivalence;
this corresponds to letting $\sim$ be the identity relation.
In \Sec{quotient} I observe that any choice of $\sim$ can be reduced to the identity relation,
namely by taking as values $\sim$-equivalence classes of values.
This reduction preserves the five sanity requirements.

After these preparations, \Sec{pg semantics} defines the promised process graph interpretation of
TSSs. Some TSSs do not have a process graph interpretation, but I show that the large class of
\emph{pure} TSSs do.

Semantic equivalences on process graphs can be lifted to open terms conform either the closed-term
or the process graph interpretation. \Sec{lifting} shows, under some mild conditions, that for pure TSSs
the latter is more discriminating.  \Sec{congruence} illustrates on a practical process algebra that
whether a semantic equivalence is a congruence may depend on which of the two interpretations is chosen.

\Sec{relating} proves the promised result that when four of the five sanity requirements have been established
for the closed-term interpretation of a TSS, they also hold for its process graph interpretation.
It also shows that the remaining requirement almost always holds

\Sec{conservative} argues that something is gained by moving from the closed-term interpretation of
TSSs to the process-graph interpretation. Based on \ex{context dependence} it formulates an
intuitively plausible theorem relating relative expressiveness of specification languages and
conservative extensions, and shows how this theorem fails under the closed-term interpretation, but holds
under the process graph interpretation.

\Sec{related} addresses related work, and
\Sec{conclusion} evaluates the five sanity requirements for languages specified by TSSs of a specific form.

\section{Syntax}\label{sec:syntax}

In this paper $\Var$ is an infinite set of \phrase{variables}, ranged over
by $X,Y,x,y,x_i$ etc.

\begin{definition}{signatures} ({\em Terms}).
A \phrase{function declaration} is a pair $(f,n)$ of a \phrase{function
symbol} $f\mathbin{\not\in} \Var$ and an \phrase{arity} $n \in \IN$.\footnote{This work generalises seamlessly to
  operators with infinitely many arguments. Such operators occur, for instance, in \cite[Appendix A.2]{BrGH16b}.
  Hence one may take $n$ to be any ordinal.
  It also generalises to operators, like the \emph{summation} or \emph{choice} of CCS
  \cite{Mi90ccs}, that take any \emph{set} of arguments.} 
A function declaration $(c,0)$ is also called a \phrase{constant declaration}.
A \phrase{signature} is a set of function declarations. The set
$\IT^r(\Sigma)$ of \phrase{terms with recursion} over a signature $\Sigma$
is defined inductively by:
\begin{itemize}
\item $\Var \subseteq \IT^r(\Sigma)$,
\item if $(f,n) \in \Sigma$ and $t_1,\ldots,t_n \in \IT^r(\Sigma)$ then
$f(t_1,\ldots,t_n) \in \IT^r(\Sigma)$,
\item If $V_S \subseteq \Var$, $~S:V_S \rightarrow \IT^r(\Sigma)$ and $X\in V_S$,
then $\rec{X|S}\in \IT^r(\Sigma)$.
\end{itemize}
A term $c()$ is abbreviated as $c$. A function $S$ as appears in
the last clause is called a \phrase{recursive specification}. It is often displayed as $\{X=S_X \mid X \in V_S\}$.
Each term $S_Y$ for $Y\in V_S$ counts as a subterm of $\rec{X|S}$.
An occurrence of a variable $y$ in a term $t$ is \phrase{free} if it does not
occur in a subterm of the form $\rec{X|S}$ with $y \in V_S$. For $t\in\IT^r(\Sigma)$ a
term, $\var(t)$ denotes the set of variables occurring free in $t$. A term
is \phrase{closed} if it contains no free occurrences of variables. For $W\subseteq \Var$, let
$\IT^r(\Sigma,W)$ denote the set of terms $t$ with $\var(t)\subseteq W$,
and let $\T^r(\Sigma)=\IT^r(\Sigma,\emptyset)$ be the set of closed terms over $\Sigma$.
The  sets $\IT(\Sigma)$, $\IT(\Sigma,W)$ and $\T(\Sigma)$ of open and closed terms over
$\Sigma$ without recursion are defined likewise, but without the last clause.
\end{definition}

\begin{definition}{substitutions}({\em Substitution}).
A {\em $\Sigma$-substitution}\index{substitution} $\sigma$ is a partial function from $\Var$ to \plat{\IT^r(\Sigma)}.
It is \phrase{closed} if it is a total function from $\Var$ to \plat{\T^r(\Sigma)}. 
If $\sigma$ is a substitution and $t$ a term, then
$t[\sigma]$ denotes the term obtained from $t$ by replacing, for $x$
in the domain of $\sigma$, every free occurrence of $x$ in $t$ by
$\sigma(x)$, while renaming bound variables if necessary to prevent
name-clashes. In that case $t[\sigma]$ is called a \phrase{substitution
instance} of $t$. A substitution instance $t[\sigma]$ where $\sigma$
is given by $\sigma(x_i)=u_i$ for $i\in I$ is denoted as
$t[u_i/x_i]_{i\in I}$, and for $S$ a recursive specification $\rec{t|S}$ abbreviates $t[\rec{Y|S}/Y]_{Y\in V_S}$.
\end{definition}
Sometimes the syntax of a language is given as a signature together with an
annotation\index{annotated signature} that places some restrictions on the use of recursion \cite{vG94}.
This annotation may for instance require the sets $V_S$ to be finite, the functions $S$ computable,
or the sets of equations $S$ to be \emph{guarded}: a syntactic criterion that ensures that
they have unique solutions under a given interpretation.
It may also rule out recursion altogether.

\section{Semantics}\label{sec:semantics}

A \phrase{language} can be given by an annotated signature, specifying its syntax,
and an \phrase{interpretation}, assigning to every term $t$ its
meaning $\den{t}$. The meaning of a closed term is a \phrase{value} chosen from a
class of values $\ID$, which is called a \phrase{domain}.
The meaning of an open term is a \emph{$\Var\!$-ary operation} on $\ID$: a function of
type $\ID^\Var\!\rightarrow\ID$, where $\ID^\Var$ is the class of functions from $\Var$ to $\ID$.
It associates a value $\den{t}(\rho)\mathbin\in\ID$ to $t$ that\linebreak[2]
depends on the choice of a \phrase{valuation} $\rho\!:\Var\rightarrow\ID$.
The valuation assigns a value from $\ID$ to each variable. 

A \phrase{partial valuation} is a function $\xi :W \rightarrow \ID$ for $W\subseteq \Var$
that assigns a value only to certain variables. 
A \phrase{$W\!$-ary operation}, for $W \subseteq \Var$, is a function $F:\ID^W
\rightarrow \ID$. It associates a value to every $W$-tuple of values,
i.e.\ to every partial valuation of the variables with domain $W$.
If $F:\ID^\Var \rightarrow \ID$ and $\zeta \in \ID^{\Var\setminus W}$ then
$F(\zeta) : \ID^W \rightarrow \ID$ is given by $F(\zeta)(\xi) =
F(\xi\cup\zeta)$ for any $\xi \in \ID^W$.  For $\rho$ a valuation and
$W$ a set of variables, $\rho\backslash W$ is the partial valuation
with domain $\Var\backslash W$ such that $(\rho\backslash W) (x) = \rho (x)$ for $x
\in \Var\backslash W$.

\subsection{Sanity requirements on interpretations}\label{sec:requirements}

Usually interpretations are required to satisfy some sanity requirements.
The work \cite{vG94} proposed five such requirements: \emph{compositionality}, applied to variables,
$n$-ary operators and recursion, \emph{invariance under $\alpha$-conversion}, and the
\emph{recursive definition principle} (RDP).

In this paper I work with domains of interpretation $\ID$ that are equipped with a semantic
equivalence relation ${\sim} \subseteq \ID\times\ID$. It indicates that values $\val,\wal\in\ID$ with $\val \sim \wal$ need
not be distinguished on our chosen level of abstraction. The equivalence $\sim$ extends to functions
$F,G\!:\!\ID^W\!\!\rightarrow\ID$ by $F\mathbin\sim G$ iff $F(\xi)\mathbin\sim G(\xi)$ for all $\xi\mathbin\in \ID^W\!$.
It extends to partial valuations $\rho,\nu: W\rightarrow\ID$ or functions $\rho,\nu$ of type
$(\ID^W\!\!\rightarrow\ID)^W$ by $\rho\sim\nu$ iff $\rho(X)\sim\nu(X)$ for all $X\in W$.
Such an equivalence relation relaxes the requirements invariance under $\alpha$-conversion and RDP,
and modifies compositionality; I speak of \emph{compositionality up to $\sim$}.
The default case in which no semantic equivalence is in force
corresponds to taking $\sim$ to be the identity relation.

\emph{Compositionality up to $\sim$\index{compositionality}} demands
that the meaning of a variable is given by the chosen valuation, i.e.,\vspace{-1ex}
\begin{equation}\label{variable interpretation}
\den{x}(\rho) \sim \rho(x)
\end{equation}
for each $x\in \Var$ and valuation $\rho:\Var\rightarrow\ID$,
and that the meaning of a term is completely determined by the meaning of its direct subterms.
This means that for operators $(f,n)\in \Sigma$ and valuations
$\rho,\nu: \Var\rightarrow\ID$\vspace{-.5pt}
\begin{equation}\label{comp-operators}
\den{t_i}(\rho)\sim \den{u_i}(\nu) ~(\mbox{for all}~i=1,...,n) ~~\Rightarrow~~
\den{f(t_1,...,t_n)}(\rho) \sim \den{f(u_1,...,u_n)}(\nu)
\vspace{-.5pt}
\end{equation}
and for recursive specifications $S$ and $S'$ with $X\in V_S = V_{S'}$
and valuations $\rho,\nu: \Var\rightarrow\ID$\vspace{-.5pt}
\begin{equation}\label{comp-recursion}
\den{S_Y}(\rho\backslash V_S) \sim \den{S'_Y}(\nu\backslash V_S) ~(\mbox{for all}~Y\in V_S)
~~\Rightarrow~~ \den{\rec{X|S}}(\rho) \sim \den{\rec{X|S'}}(\nu).
\vspace{-.5pt}
\end{equation}
Note that the precondition of (\ref{comp-recursion}) evaluates the
variables in $S_Y$ and $S'_Y$ that are free in $\rec{X|S}$ and $\rec{X|S'}$
according to $\rho$ and $\nu$, respectively, and the variables from
$V_S=V_{S'}$ under any common valuation $\zeta$.

\emph{Invariance under $\alpha$-conversion}\index{invariance under $\alpha$-conversion} demands that the meaning of a term is
independent of the names of its bound variables, i.e.\
for any injective substitution $\gamma:V_S \rightarrow \Var$ such that the range of $\gamma$ contains
no variables occurring free in $\rec{S_Y|S}$ for some $Y\in V_S$\vspace{-.5pt}
\begin{equation}\label{alpha-conversion1}
\den{\rec{\gamma(X) | S[\gamma]}} \sim \den{\rec{X|S}}.
\end{equation}
Finally, the meaning of a term $\rec{X|S}$ should be
the $X$-component of a solution of $S$. To be precise,
\begin{equation}\label{solutions}
\den{\rec{X|S}} \sim \den{\rec{S_X|S}}.
\end{equation}
This property is called the \phrase{recursive definition principle} \cite{Fok00}.

A straightforward structural induction on $t$ using
(\ref{variable interpretation}), (\ref{comp-operators}) and
(\ref{comp-recursion}) shows that $\den{t}(\rho)$ depends only on the
restriction of $\rho$ to those variables that occur free in $t$.
If there are no such variables, $\den{t}(\rho)$ does not depends on $\rho$ at all,
and consequently can be abbreviated to $\den{t}$.

\subsection{Alternative forms of the sanity requirements}

Note that (\ref{comp-operators}) holds iff for every $(f,n)\in \Sigma$ there is
a function $f^{\IDs}:\ID^n \rightarrow \ID$ such that 
$$\forall \rho\in \ID^\Var: \den{f(t_1,\ldots,t_n)} (\rho) \sim
f^{\IDs}(\den{t_1}(\rho),\ldots,\den{t_n}(\rho)).$$ 
Requirement (\ref{comp-recursion}) can be characterised in the same vein:
\begin{proposition}{comp. recursion}
Property (\ref{comp-recursion}) holds iff for
every set $W \subseteq \Var$ there is a function $\mu_W^{\IDs}:
(\ID^W\! \rightarrow\ID)^{W} \rightarrow \ID^W$ such that for every
recursive specification $S:W \rightarrow \IT^r(\Sigma)$ with $X\in W$, and every $\rho:\Var\rightarrow\ID$,
$$\den{\rec{X|S}}(\rho) \sim \mu_W^{\IDs} (\den{S}(\rho\backslash W))(X).$$
\end{proposition}
\begin{proof}
  Since the meaning of term $S_Y\in\IT^r(\Sigma)$ is of type $\ID^\Var\rightarrow\ID$,
the meaning of a recursive specification $S:W\rightarrow \IT^r(\Sigma)$ is of type
$\den{S}: W\rightarrow (\ID^\Var\rightarrow\ID)$. Applying to that a partial valuation
$\rho\backslash W: \Var\backslash W \rightarrow \ID$ yields a function $\den{S}(\rho\backslash W)$
of type $(\ID^W\! \rightarrow\ID)^{W}$. Now for each $\chi\in (\ID^W\! \rightarrow\ID)^{W}$ choose,
if possible, a pair $S^\chi:W\rightarrow\IT^r(\Sigma)$ and $\rho^\chi:\Var\rightarrow\ID$ such that
$\den{S^\chi}(\rho^\chi\backslash W) \sim \chi$, and define $\mu_W^{\IDs} (\chi)(X)$, for $X\in W$,
to be $\den{\rec{X|S^\chi}}(\rho^\chi)$. For a $\chi$ for which no such pair can be found the
definition of $\mu_W^{\IDs} (\chi)(X)$ is arbitrary. Now pick $S:W \rightarrow \IT^r(\Sigma)$
and $\rho:\Var\rightarrow\ID$. Let $\chi:= \den{S}(\rho\backslash W)$.
Then $\den{S}(\rho\backslash W) = \chi \sim \den{S^\chi}(\rho^\chi\backslash W)$, so by 
(\ref{comp-recursion}) one has $$\den{\rec{X|S}}(\rho) \sim \den{\rec{X|S^\chi}}(\rho^\chi) =
\mu_W^{\IDs} (\chi)(X) = \mu_W^{\IDs}(\den{S}(\rho\backslash W))(X).$$
The other direction, that the existence of such a $\mu_W^{\IDs}$ implies (\ref{comp-recursion}), is trivial.
\end{proof}

\noindent
Write \plat{t \eqa u} if the terms $t,u\in\IT^r(\Sigma)$ differ only in the
names of their bound variables. Then (\ref{alpha-conversion1}) can be
rewritten as\vspace{-1ex}
\begin{equation}\label{alpha-conversion}
t\eqa u ~~\Rightarrow~~ \den{t} \sim \den{u}.
\vspace{-2ex}
\end{equation}
\begin{proposition}{alpha-conversion}
In the presence of (\ref{comp-operators}) and (\ref{comp-recursion}),
(\ref{alpha-conversion1}) is equivalent to (\ref{alpha-conversion}).
\end{proposition}
\begin{proof}
Clearly, (\ref{alpha-conversion1}) is a special case of (\ref{alpha-conversion}).
The other direction proceeds by structural induction on $t$.

In case $t=X\in \Var$ then $u=X$ and thus $\den{t} \sim \den{u}$.

Let $t=f(t_1,...,t_n)$. Then $u=f(u_1,...,u_n)$ and $t_i\eqa u_i$ for each $i=1,...,n$.
By induction $\den{t_i} \sim \den{u_i}$ for each $i$. This means that
$\den{t_i}(\rho) \sim \den{u_i}(\rho)$ for each $\rho:\Var\rightarrow\ID$.
Hence, by (\ref{comp-operators}), $\den{t} \sim \den{u}$.

Let $t=\rec{X|S}$. Then $u=\rec{\gamma(X) | S'[\gamma]}$ for a recursive specification
$S':V_S\rightarrow\IT^r(\Sigma)$ with $S_Y \eqa S'_Y$ for all $Y\in V_S$,
and an injective substitution $\gamma:V_S \rightarrow \Var$ such that the range of $\gamma$ contains
no variables occurring free in $\rec{S'_Y|S}$ for some $Y\in V_S$.
By induction $\den{S_Y} \sim \den{S'_Y}$ for each $Y$. This means that
$\den{S_Y}(\rho) \sim \den{S'_Y}(\rho)$ for each $\rho:\Var\rightarrow\ID$,
so in particular $\den{S_Y}(\rho\backslash V_S) \sim
\den{S'_Y}(\rho\backslash V_S)$ for each such $\rho$.
Hence, by (\ref{comp-recursion}), $\den{\rec{X|S}} \sim \den{\rec{X|S'}}$.
Thus, by (\ref{alpha-conversion1}), $\den{t} \sim \den{\rec{X|S'}} \sim
\den{\rec{\gamma(X) | S'[\gamma]}} =\den{u}$.
\end{proof}

\subsection{Applying semantic interpretations to substitutions}\label{sec:substitutions}

The semantic mapping $\den{\ \ }:\IT^r(\Sigma) \rightarrow ((\Var\rightarrow\ID)\rightarrow\ID)$
extends to substitutions $\sigma:\Var\rightharpoonup\IT^r(\Sigma)$ by
$\den{\sigma}(\rho)(X):=\den{\sigma(X)}(\rho)$
for all $X\mathbin\in \Var$ and $\rho:\Var\rightarrow\ID$---here $\sigma$ is extended to a total
function by $\sigma(Y):=Y$ for all $Y\not\in\dom(\sigma)$. Thus $\den{\sigma}$ is of type
$(\Var\rightarrow\ID)\rightarrow(\Var\rightarrow\ID)$, i.e.\ a map from valuations to valuations.
The following results applies to languages satisfying sanity requirements
(\ref{variable interpretation})--(\ref{alpha-conversion1}).

\begin{proposition}{inductive meaning}
Let $t\in\IT^r(\Sigma)$ be a term, $\sigma:\Var\rightharpoonup\IT^r(\Sigma)$ a substitution, and
$\rho:\Var\rightarrow\ID$ a valuation. Then\vspace{-2ex}
\begin{equation}\label{inductive meaning}
\den{t[\sigma]}(\rho) \sim \den{t}(\den{\sigma}(\rho)).\vspace{1ex}
\end{equation}
\end{proposition}

\begin{proof}
By the definition of substitution, there is an $u\in\IT^r(\Sigma)$
with $t\eqa u$, such that $t[\sigma]=u[\sigma]$, and when performing
the substitution $\sigma$ on $u$ there is no need to rename any bound
variables occurring in $u$.\vspace{-2pt} It now suffices to obtain
$\den{u[\sigma]}(\rho) \sim \den{u}(\den{\sigma}(\rho))$,
because then $\den{t[\sigma]}(\rho) = \den{u[\sigma]}(\rho) \sim
\den{u}(\den{\sigma}(\rho)) \mathrel{{\stackrel{\mbox{\scriptsize (\ref{alpha-conversion})}}\sim}}
\den{t}(\den{\sigma}(\rho))$.\pagebreak[3]
For this reason it suffices to establish (\ref{inductive meaning}) for
terms $t$ and substitutions $\sigma$ with the property (*) that
whenever a variable $Z$ occurs free within a subterm $\rec{X|S}$ of
$t$ with $Y\in V_S$ then $Y$ does not occur free in $\sigma(Z)$.
I proceed with structural induction on $t$, while quantifying over all $\rho$.

Let $t=X\in \Var$. Then
$\den{X}(\den{\sigma}(\rho)) \stackrel{\mbox{\scriptsize (\ref{variable interpretation})}}\sim
(\den{\sigma}(\rho))(X) \stackrel{\mbox{\scriptsize\it def}}=
\den{\sigma(X)}(\rho) =\den{X[\sigma]}(\rho)$.

Let $t=f(t_1,...,t_n)$. By induction I may assume that
$\den{t_i[\sigma]}(\rho) \sim \den{t_i}(\den{\sigma}(\rho))$ for $i=1,...,n$.
Hence\vspace{-5pt}
$$\den{f(t_1,\ldots,t_n)[\sigma]}(\rho) =
\den{f(t_1[\sigma],\ldots,t_n[\sigma])}(\rho) \stackrel{\mbox{\scriptsize (\ref{comp-operators})}}\sim
\den{f(t_1,\ldots,t_n)}(\den{\sigma}(\rho)).$$

Let $t=\rec{X|S}$. Then $t[\sigma] = \rec{X|S[\sigma\backslash V_S]}$.
Given that the pair $t,\sigma$ satisfies property (*), so do the pairs
$S_Y,\sigma\backslash V_S$ for all $Y\in V_S$. Moreover, no $Y\in V_S$
occurs free in $\sigma(Z)$ for $Z\notin V_S$ occurring free in $S$.\hfill (\#)\newline
By induction I assume $\den{S_Y[\sigma\backslash V_S]}(\rho) \sim
\den{S_Y}\big(\den{\sigma\backslash V_S}(\rho)\big)$ for all $Y\in V_S$ and
all $\rho:\Var\rightarrow\ID$. Any such $\rho$ can be written as
$(\rho\backslash V_S)\cup \xi$ for some $\xi:V_S\rightarrow \ID$.
Now (\#) yields that for all $Z\mathbin\in \Var$ occurring free in $S$\vspace{-1ex}
$$\big( \den{\sigma\backslash V_S}((\rho\backslash V_S)\cup \xi)\big)(Z) =
 \big(((\den{\sigma}(\rho))\backslash V_S)\cup \xi\big)(Z).$$
Hence $\den{S_Y[\sigma\backslash V_S]}(\rho\backslash V_S)(\xi) \sim
\den{S_Y}\big(\den{\sigma\backslash V_S}((\rho\backslash V_S)\cup \xi)\big) =
\den{S_Y}\big((\den{\sigma}(\rho))\backslash V_S\big)(\xi)$
for all $Y\in V_S$, $\rho\!:\Var\rightarrow\ID$ and $\xi\!:V_S\rightarrow\ID$.
So $\den{S_Y[\sigma\backslash V_S]}(\rho\backslash V_S) \sim
\den{S_Y}\big((\den{\sigma}(\rho))\backslash V_S\big)$
for all $Y\mathbin\in V_S$ and $\rho\!:\Var\rightarrow\ID$.
One obtains\hfill
$\den{\rec{X|S}[\sigma]}(\rho)
= \den{\rec{X|S[\sigma\backslash V_S]}}(\rho)
\stackrel{\mbox{\scriptsize (\ref{comp-recursion})}}\sim
\den{\rec{X|S}}(\den{\sigma}(\rho))$.
\end{proof}

\section{Quotient Domains}\label{sec:quotient}

An equivalence relation $\sim$ on $\ID$ is a \phrase{congruence}%
\footnote{This property is called \phrase{lean congruence} in \cite{vG17b}.
There $\sim$ is called a \phrase{full congruence} for $\fL$ iff $\fL$ is compositional up to $\sim$.
In the absence of recursion this is equivalent to (\ref{congruence}), but in general it is
a stronger requirement---i.e., the reverse of \pr{congruence} does not hold \cite{vG17b}.} for $\fL$ if
\begin{equation}\label{congruence}
\rho\sim\nu  ~~\Rightarrow~~ \den{t}(\rho)\sim\den{t}(\nu)
\end{equation}
for any term $t$ and any valuations $\rho,\nu:\Var\rightarrow\ID$.

\begin{proposition}{congruence}
If a language $\fL$ is compositional up to an equivalence
$\sim$ then $\sim$ is a congruence for $\fL$.
\end{proposition}

\begin{proof}
A straightforward structural induction on $t$.
\end{proof}

Given a domain $\ID$ for interpreting languages and an equivalence
relation $\sim$, the \phrase{quotient domain} $\ID/_\sim$ consists of the
$\sim$-equivalence classes of elements of $\ID$. For $\val\in\ID$
let $[\val]_\sim\in\ID/_\sim$ denote the equivalence class containing $\val\in\ID$.
Likewise, for a valuation $\rho:\Var\rightarrow\ID$ in $\ID$, the valuation
\plat{[\rho]_\sim:\Var\rightarrow\ID/_\sim} in $\ID/_\sim$ is given by $[\rho]_\sim(x):= [\rho(x)]_\sim$;
it also represents the $\sim$-equivalence class of valuations in $\ID$
of which $\rho$ is a member.
Each valuation in $\ID/_\sim$ is of the form $[\rho]_\sim$.

An interpretation $\den{\_}: \IT^r(\Sigma) \rightarrow (\ID^\Var\rightarrow\ID)$
that satisfies (\ref{congruence}) is turned into the \phrase{quotient interpretation}
$\den{\_}_\sim: \IT^r(\Sigma) \rightarrow ((\ID/_\sim)^\Var\rightarrow\ID/_\sim)$
by defining $\den{t}_\sim([\rho]_\sim):=[\den{t}(\rho)]_\sim$.
By (\ref{congruence}), this is independent of the choice of a
representative valuation $\rho$ within the equivalence class $[\rho]_\sim$.

Let $\den{\_}$ be an interpretation and $\sim$ an equivalence such that (\ref{congruence}) holds.
Then $\den{\_}$ satisfies  the sanity requirements (\ref{variable interpretation})--(\ref{solutions})
of Section~\ref{sec:semantics} up to $\sim$ iff $\den{\_}_\sim$ satisfies these requirements up to $=$.

\section{Transition System Specifications}\label{sec:TSS}

\begin{definition}{TSS} (\emph{Transition system specification;}
{\sc Groote \& Vaandrager} \cite{GrV92}).
Let $\Sigma$ be an annotated signature and $A$ a set (of \phrase{actions}).
A {\em (positive) $(\Sigma,A)$-literal}\index{literal} is an expression
\plat{t \ar{a} t'} with \plat{t,t'\in \IT^r(\Sigma)} and $a \in A$.
A \phrase{transition rule} over $(\Sigma,A)$ is an expression of the
form \plat{\frac{H}{\lambda}} with $H$ a set of $(\Sigma,A)$-literals (the
\phrase{premises} of the rule) and $\lambda$ a $(\Sigma,A)$-literal (the
\phrase{conclusion}). A rule \plat{\frac{H}{\lambda}} with $H\mathbin=\emptyset$ is
also written $\lambda$.\linebreak[3] A \phrase{transition system
specification (TSS)} is a triple $(\Sigma,A,R)$ with
$R$ a set of transition rules over $(\Sigma,A)$.
\end{definition}
The following definition (from \cite{vG93d}) tells when a literal
is provable from a TSS\@. It generalises the standard definition (see
e.g.\ \cite{GrV92}) by (also) allowing the derivation of transition rules.
The derivation of a literal \plat{t\ar{a}t'} corresponds to the
derivation of the rule \plat{\frac{H}{t\ar{\scriptscriptstyle a}t'}} with $H=\emptyset$.
The case $H \neq \emptyset$ corresponds to the derivation of
\plat{t\ar{a}t'} under the assumptions $H$.

\begin{definition}{proof}({\em Proof\/}).
Let $P=(\Sigma,A,R)$ be a TSS. A \phrase{proof} of a transition rule
\plat{\frac{H}{\lambda}} from $P$ is a well-founded, upwardly
branching tree of which the nodes are labelled by $(\Sigma,A)$-literals,
such that:
\begin{itemize}
\item the root is labelled by $\lambda$, and
\item if $\kappa$ is the label of a node $q$ and $K$ is the set of
labels of the nodes directly above $q$, then
\begin{itemize}
\item either $K=\emptyset$ and $\kappa \in H$,
\item or \plat{\frac{K}{\kappa}} is a substitution instance of a rule from $R$.
\end{itemize}
\end{itemize}
If a proof of $\frac{H}{\lambda}$ from $P$ exists, then
$\frac{H}{\lambda}$ is \phrase{provable} from $P$, denoted $P \vdash \frac{H}{\lambda}$.
\end{definition}
A \phrase{labelled transition system} (LTS) is a triple $(S,A,\rightarrow)$ with
$S$ a set of \phrase{states} or \phrase{processes}, $A$ a set of
\phrase{actions}, and $\mathord\rightarrow \subseteq S \times A \times S$ the 
\phrase{transition relation}, or set of \phrase{transitions}.
A TSS $P=(\Sigma,A,R)$ \emph{specifies} the LTS
$(\T^r(\Sigma),A,\rightarrow)$ whose states are the closed terms over $\Sigma$
and whose transitions are the closed literals provable from $P$.

For the sake of simplicity, the above treatment of TSSs deals with positive premises only.
However, all results of this paper apply equally well, and with unaltered proofs, to TSS with
negative premises $t \!\nar{a}$, following the treatment below. The rest of the section may be
skipped in first reading.

\subsection{TSSs with negative premises}
A \emph{negative} $(\Sigma,A)$-\emph{literal} is an expression \plat{t\nar a}.
A transition rule may have positive and negative literals as premises, but must have a positive conclusion.
Literals \plat{t \ar{a} u} and \plat{t\nar{a}} are said to \emph{deny} each other.

\begin{definition}{wsp} \cite{vG04}
  Let $P=(\Sigma,A,R)$ be a TSS\@. A
  \emph{well-supported proof} from $P$ of a closed literal ${\lambda}$
  is a well-founded tree with the nodes labelled by closed literals,
  such that the root is labelled by $\lambda$, and if $\kappa$ is the
  label of a node and $K$ is the set of labels of the children of this
  node, then:
  \begin{enumerate}
  \item either $\kappa$ is positive and \plat{\frac{K}{\kappa}} is a
    closed substitution instance of a rule in $R$;
  \item or $\kappa$ is negative and for each set $N$ of closed negative
    literals with \plat{\frac{N}{\nu}} provable from
    $P$ and $\nu$ a closed positive literal denying $\kappa$, a literal
    in $K$ denies one in $N$.
  \end{enumerate}
  $P\vdash_{\it ws}\lambda$ denotes that a well-supported proof from
  $P$ of $\lambda$ exists.  A standard TSS $P$ is \emph{complete} if
  for each $p$ and $a$, either
  \plat{P\vdash_{\it ws}p\nar a} or there exists a closed
  term $q$ such that \plat{P\vdash_{\it ws}p\ar a q}.
\end{definition}
In \cite{vG04} it is shown that no TSS admit well-supported proofs of literals that deny each other.
Only a complete TSSs specifies an LTS;
its transitions are the closed positive literals with a well-supported proof.

\section{The Closed-term Semantics of Transition System Specifications}\label{sec:closed-term}

\hypertarget{closed-term}{
The default semantics of a language given as a TSS $(\Sigma,A,R)$ is
to take the domain $\ID$ in which the expressions are interpreted to
be $\T^r(\Sigma)$, the set of closed terms over $\Sigma$.
The meaning of a closed expression $p\in \T^r(\Sigma)$ is simply itself: $\den{p}:=p$.
The meaning $\den{t}\in \ID^\Var\rightarrow\ID$ of an open expression $t\in \IT^r(\Sigma)$
is given by $\den{t}(\rho):=t[\rho]$. Here one uses the fact that a valuation
$\rho:\Var\rightarrow\ID$ is also a closed substitution $\rho:\Var\rightarrow\T^r(\Sigma)$.
Given a semantic equivalence relation ${\sim}\subseteq \T^r(\Sigma)\times \T^r(\Sigma)$,
the closed-term semantics of a TSS always satisfies Requirement~(\ref{variable interpretation}),
whereas (\ref{comp-operators})--(\ref{solutions}) simplify to\vspace{-1ex}
\begin{equation}\tag{$2'$}\label{comp-operators-closed}
p_i\sim q_i ~(\mbox{for all}~i=1,...,n) ~~\Rightarrow~~ f(p_1,...,p_n) \sim f(q_1,...,q_n)~~~~~\mbox{and}
\vspace{-4ex}
\end{equation}
\begin{equation}\tag{$3'$}\label{comp-recursion-closed}
S_Y[\sigma]\sim S_Y'[\sigma] ~(\mbox{for all}~Y\in W~\mbox{and}~ \sigma:W \rightarrow \T^r(\Sigma))
~~\Rightarrow~~ \rec{X|S} \sim \rec{X|S'}
\vspace{-3ex}
\end{equation}
\begin{equation}\tag{$4'$}\label{alpha-conversion1-closed}
\rec{\gamma(X) | S[\gamma]} \sim \rec{X|S}
\vspace{-3ex}
\end{equation}
\begin{equation}\tag{$5'$}\label{solutions-closed}
\rec{X|S} \sim \rec{S_X|S}
\end{equation}
for all functions $(f,n)\in\Sigma$, closed terms $p_i,q_i\in \T^r(\Sigma)$,
recursive specifications $S,S':W \rightarrow \IT^r(\Sigma,W)$ with $X \in W \subseteq \Var$,
and $\gamma:W \rightarrow \Var$ injective.
}

\section{Process Graphs}\label{sec:process graphs}

When the expressions in a language are meant to represent processes,
they are called \phrase{process expressions}, and the language a
\phrase{process description language}.
Suitable domains for interpreting process description languages are
the class of {\em process graphs} \cite{BK84} and its quotients. In such
\phrase{graph domains} a process is represented by either a process graph, or
an equivalence class of process graphs.  Process graphs are also known
as \phrase{state-transition diagrams} or \phrase{automata}.  They are
LTSs equipped with an initial state.
A process graph can also be seen as a state in an LTS\@.

\begin{definition}{graphs}
A \phrase{process graph}, labelled over a set $A$ of actions, is a
triple $G=(S,A,\rightarrow,I)$ with
\begin{itemize}
\item[\bf --] $S$ a set of \emph{nodes} or \phrase{states},
\item[\bf --] $\mathord\rightarrow \subseteq S \times A \times S$ a set of \emph{edges}
	      or \phrase{transitions},
\item[\bf --] and $I \in S$ the \phrase{root} or {\em initial} state\index{initial state}.
\end{itemize}
Let $\IG(A)$ be the domain of process graphs labelled over $A$.
\end{definition}
One writes $r\ar{a}s$ for $(r,a,s)\in\mathord{\rightarrow}$.
Virtually all so-called {\em interleaving models} for the
representation of processes are isomorphic to graph models. For instance,
the {\em failure sets} that represent expressions in the process
description language CSP \cite{BHR84} can easily be encoded as
equivalence classes of graphs, under a suitable equivalence.
In \cite{BK84} the language ACP is equipped with a process graph
semantics, and the semantics of CCS, SCCS and {\sc Meije} given in
\cite{Mi90ccs,Mi83,AB84,dS85} are operational ones, which, as I will
show below, induce process graph semantics.
In the languages $\fL$ studied in this paper, the domain $\ID$ in
which $\fL$-expressions are interpreted will be $\IG(A)$ for some set
of actions $A$, or a subclass of $\IG(A)$.

Usually the parts of a graph that cannot be reached from the initial
state by following a finite path of transitions are considered
meaningless for the description of processes. This means that one is
only interested in process graphs as a model of system behaviour up to
some equivalence, and this equivalence identifies at least graphs with
the same reachable parts.

\begin{definition}{reachable}
The \phrase{reachable part} of a process graph $(S,A,\rightarrow,I)$
is the process graph $(S',A,\rightarrow',I)$ where
\begin{itemize}
\item $S'\subseteq S$ is the smallest
set such that (1) $I\in S'$ and (2) if $r\in S'$ and $r\ar{a}s$ then $s\in S'$,
\item and $\rightarrow'$ is the restriction of $\rightarrow$ to $S'\times A\times S'$.
\end{itemize}
\end{definition}

\section{A Process Graph Semantics of Transition System Specifications}\label{sec:pg semantics}

This section proposes a process graph semantics of TSSs.
For each TSS $P=(\Sigma,A,R)$ it defines an interpretation $\den{\_}_P : \IT^r(\Sigma) \rightarrow (\IG(A)^\Var\rightarrow\IG(A))$,
thus taking $\ID$ to be $\IG(A)$.

\begin{definition}{interpreting}({\em Interpreting
the closed expressions in a TSS as process graphs}).
Let $P=(\Sigma,A,R)$ be a TSS and $p \in \T^r(\Sigma)$. 
Then $\den{p}_P^\emptyset\in\IG(A)$ is
the reachable part of the process graph $(\T^r(\Sigma),A,\rightarrow,p)$
with $\rightarrow$ the set of transitions provable from $P$.
\end{definition}
To define an interpretation $\den{\_}_P : \IT^r(\Sigma) \rightarrow (\IG(A)^\Var\rightarrow\IG(A))$
of the open $\Sigma$-terms in $\IG(A)$ I would like to simply add to the signature $\Sigma$ 
a constant $G$ for each process graph $G\in\IG(A)$. However, $\IG(A)$
is a proper class, whereas a signature needs to be a set. For this reason I work
with appropriate subsets $\IG^*$ of $\IG(A)$ instead of with $\IG(A)$ itself.
I will discuss the selection of $\IG^*$ later, but one requirement will be\\[1ex]
\mbox{}\hfill if $(S,A,\rightarrow,r)\in \IG^*$ and $r\ar{a}s$ then also $(S,A,\rightarrow,s)\in \IG^*$
\hfill (\phrase{transition closure}).\\[1ex]
Define the transition relation $\mathord{\rightarrow_{\IGs^*}} \subseteq \IG^* \times A \times \IG^*$
by $G\ar{a}_{\IGs^*}G'$ iff
(i) $G=(S,A,\rightarrow,r)$, (ii) there is a transition $(r,a,s)\mathbin\in \mathord{\rightarrow}$, and
(iii) $G' = (S,A,\rightarrow,s)$ is the same graph but with $s$ as initial state.

Now consider a term $t\in\IT^r(\Sigma)$ and a valuation $\rho:\Var\rightarrow\IG(A)$.
In order to define $\den{t}_P(\rho)$, make sure that $\IG^*$ \phrase{supports} $(t,\rho)$, meaning
that it contains $\rho(x)$ for any variable $x\in \Var$ occurring free in $t$.  Let \plat{P+\IG^*}
be the TSS $P$ to which all graphs \plat{G\in\IG^*} have been added as
constants, and all transitions in \plat{\rightarrow_{\IGs^*}} as
transition rules without premises.
As the valuation $\rho$ now also is a substitution, $t[\rho]$ is a closed term
in the TSS \plat{P+\IG^*}.
Define \plat{\den{t}_P^{\IGs^*}(\rho)} to be $\den{t[\rho]}_{P+\IGs^*}^\emptyset \in \IG(A)$: the interpretation
according to \df{interpreting} of the closed term $t[\rho]$ from the TSS $P+\IG^*$.

\begin{definition}{simply interpreting}({\em The simple process graph semantics of TSSs}).
A TSS $P$ \phrase{manifestly induces a process graph semantics} iff, for any term $t$ and valuation $\rho$,
the interpretation \plat{\den{t}_P^{\IGs^*}(\rho)} is independent of the choice of $\IG^*\!$, as
long as $\IG^*$ is transition-closed and supports $(t,\rho)$.
In that case $\den{t}_P(\rho)$ is defined to be $\den{t}_P^{\IGs^*}(\rho)$.
\end{definition}
It is possible to enlarge the class of TSSs that induce a process graph semantics a little bit:
\begin{definition}{strongly interpreting}({\em The process graph semantics of TSSs}).
Call a choice of $\IG^*$ \phrase{adequate} for (the interpretation of)
$\den{t}_P(\rho)$ if $\IG^*$ is transition-closed and supports $(t,\rho)$, and
\plat{\den{t}_P^{\IGs'}(\rho) = \den{t}_P^{\IGs^*}(\rho)} for any transition-closed superset $\IG'$ of $\IG^*$.
Now $P$ \phrase{induces a process graph semantics} iff, for any term $t$ and valuation $\rho$,
an adequate choice $\IG^*$ for $\den{t}_P(\rho)$ exists.
In that case $\den{t}_P(\rho)$ is defined to be \plat{\den{t}_P^{\IGs^*}(\rho)}
for any adequate choice of $\IG^*$.
\end{definition}

\begin{trivlist} \item[\hspace{\labelsep}{\bf \ex{context dependence} (continued)}]
Reconsider the operators $f$ and $id$ from \ex{context dependence}. To judge whether they are
essentially different one compares the open terms $f(x)$ and $id(x)$. Their meanings are values that
depend on the choice of a valuation $\rho$, mapping variables to values. In fact they depend on the
value $\rho(x)$ only.

Under the closed term interpretation of the TSS $P$ of \ex{context dependence},
$\rho(x)$ is a closed term in the language; it cannot have an outgoing $\tau$-transition.
Thus $\den{f(x)}_P(\rho) \bis{} \den{x}_P(\rho) \bis{} \den{id(x)}_P(\rho)$ for any such $\rho$,
so $\den{f(x)}_P \bis{} \den{x}_P \bis{} \den{id(x)}_P$, i.e., $f$ and $id$ are strongly bisimilar.

However, under the process graph interpretation, $\rho(x)$ is a process graph, and
one may take $\rho(x)$ to be
\expandafter\ifx\csname graph\endcsname\relax
   \csname newbox\expandafter\endcsname\csname graph\endcsname
\fi
\ifx\graphtemp\undefined
  \csname newdimen\endcsname\graphtemp
\fi
\expandafter\setbox\csname graph\endcsname
 =\vtop{\vskip 0pt\hbox{%
    \special{pn 8}%
    \special{ar 197 39 39 39 0 6.28319}%
    \special{sh 1.000}%
    \special{pn 1}%
    \special{pa 79 14}%
    \special{pa 157 39}%
    \special{pa 79 64}%
    \special{pa 79 14}%
    \special{fp}%
    \special{pn 8}%
    \special{pa 0 39}%
    \special{pa 79 39}%
    \special{fp}%
    \special{ar 591 39 39 39 0 6.28319}%
    \special{sh 1.000}%
    \special{pn 1}%
    \special{pa 472 14}%
    \special{pa 551 39}%
    \special{pa 472 64}%
    \special{pa 472 14}%
    \special{fp}%
    \special{pn 8}%
    \special{pa 236 39}%
    \special{pa 472 39}%
    \special{fp}%
    \graphtemp=\baselineskip
    \multiply\graphtemp by -1
    \divide\graphtemp by 2
    \advance\graphtemp by .5ex
    \advance\graphtemp by 0.039in
    \rlap{\kern 0.394in\lower\graphtemp\hbox to 0pt{\hss $\tau$\hss}}%
    \special{ar 984 39 39 39 0 6.28319}%
    \special{sh 1.000}%
    \special{pn 1}%
    \special{pa 866 14}%
    \special{pa 945 39}%
    \special{pa 866 64}%
    \special{pa 866 14}%
    \special{fp}%
    \special{pn 8}%
    \special{pa 630 39}%
    \special{pa 866 39}%
    \special{fp}%
    \graphtemp=\baselineskip
    \multiply\graphtemp by -1
    \divide\graphtemp by 2
    \advance\graphtemp by .5ex
    \advance\graphtemp by 0.039in
    \rlap{\kern 0.787in\lower\graphtemp\hbox to 0pt{\hss $c$\hss}}%
    \hbox{\vrule depth0.079in width0pt height 0pt}%
    \kern 1.024in
  }%
}%

\,\raisebox{11.6pt}[0pt][0pt]{\box\graph}\;,
where the short arrow indicates the initial state.
With this valuation, the process graph $\den{id(x)}_P(\rho)$ is isomorphic to $\rho(x)$, whereas $\den{f(x)}_P$ has no
outgoing transitions. So $\den{f(x)}_P(\rho) \nobis \den{x}_P(\rho) \bis{} \den{id(x)}_P(\rho)$ and
consequently $\den{f(x)}_P \nobis \den{x}_P \bis{} \den{id(x)}_P$, i.e., $f$ and $id$ are not bisimilar.

The smallest set of process graphs $\IG^*$ that is adequate for the interpretation of $\den{f(x)}_P(\rho)$
and $\den{id(x)}_P(\rho)$ is $\left\{\rule{0pt}{12pt}
\expandafter\ifx\csname graph\endcsname\relax
   \csname newbox\expandafter\endcsname\csname graph\endcsname
\fi
\ifx\graphtemp\undefined
  \csname newdimen\endcsname\graphtemp
\fi
\expandafter\setbox\csname graph\endcsname
 =\vtop{\vskip 0pt\hbox{%
    \special{pn 8}%
    \special{ar 39 197 39 39 0 6.28319}%
    \special{sh 1.000}%
    \special{pn 1}%
    \special{pa 64 79}%
    \special{pa 39 157}%
    \special{pa 14 79}%
    \special{pa 64 79}%
    \special{fp}%
    \special{pn 8}%
    \special{pa 39 0}%
    \special{pa 39 79}%
    \special{fp}%
    \special{ar 433 197 39 39 0 6.28319}%
    \special{sh 1.000}%
    \special{pn 1}%
    \special{pa 315 172}%
    \special{pa 394 197}%
    \special{pa 315 222}%
    \special{pa 315 172}%
    \special{fp}%
    \special{pn 8}%
    \special{pa 79 197}%
    \special{pa 315 197}%
    \special{fp}%
    \graphtemp=\baselineskip
    \multiply\graphtemp by -1
    \divide\graphtemp by 2
    \advance\graphtemp by .5ex
    \advance\graphtemp by 0.197in
    \rlap{\kern 0.236in\lower\graphtemp\hbox to 0pt{\hss $\tau$\hss}}%
    \special{ar 827 197 39 39 0 6.28319}%
    \special{sh 1.000}%
    \special{pn 1}%
    \special{pa 709 172}%
    \special{pa 787 197}%
    \special{pa 709 222}%
    \special{pa 709 172}%
    \special{fp}%
    \special{pn 8}%
    \special{pa 472 197}%
    \special{pa 709 197}%
    \special{fp}%
    \graphtemp=\baselineskip
    \multiply\graphtemp by -1
    \divide\graphtemp by 2
    \advance\graphtemp by .5ex
    \advance\graphtemp by 0.197in
    \rlap{\kern 0.630in\lower\graphtemp\hbox to 0pt{\hss $c$\hss}}%
    \hbox{\vrule depth0.236in width0pt height 0pt}%
    \kern 0.866in
  }%
}%

~\,\raisebox{10.6pt}[0pt][0pt]{\box\graph}\;~,
\expandafter\ifx\csname graph\endcsname\relax
   \csname newbox\expandafter\endcsname\csname graph\endcsname
\fi
\ifx\graphtemp\undefined
  \csname newdimen\endcsname\graphtemp
\fi
\expandafter\setbox\csname graph\endcsname
 =\vtop{\vskip 0pt\hbox{%
    \special{pn 8}%
    \special{ar 39 197 39 39 0 6.28319}%
    \special{ar 433 197 39 39 0 6.28319}%
    \special{sh 1.000}%
    \special{pn 1}%
    \special{pa 458 79}%
    \special{pa 433 157}%
    \special{pa 408 79}%
    \special{pa 458 79}%
    \special{fp}%
    \special{pn 8}%
    \special{pa 433 0}%
    \special{pa 433 79}%
    \special{fp}%
    \special{sh 1.000}%
    \special{pn 1}%
    \special{pa 315 172}%
    \special{pa 394 197}%
    \special{pa 315 222}%
    \special{pa 315 172}%
    \special{fp}%
    \special{pn 8}%
    \special{pa 79 197}%
    \special{pa 315 197}%
    \special{fp}%
    \graphtemp=\baselineskip
    \multiply\graphtemp by -1
    \divide\graphtemp by 2
    \advance\graphtemp by .5ex
    \advance\graphtemp by 0.197in
    \rlap{\kern 0.236in\lower\graphtemp\hbox to 0pt{\hss $\tau$\hss}}%
    \special{ar 827 197 39 39 0 6.28319}%
    \special{sh 1.000}%
    \special{pn 1}%
    \special{pa 709 172}%
    \special{pa 787 197}%
    \special{pa 709 222}%
    \special{pa 709 172}%
    \special{fp}%
    \special{pn 8}%
    \special{pa 472 197}%
    \special{pa 709 197}%
    \special{fp}%
    \graphtemp=\baselineskip
    \multiply\graphtemp by -1
    \divide\graphtemp by 2
    \advance\graphtemp by .5ex
    \advance\graphtemp by 0.197in
    \rlap{\kern 0.630in\lower\graphtemp\hbox to 0pt{\hss $c$\hss}}%
    \hbox{\vrule depth0.236in width0pt height 0pt}%
    \kern 0.866in
  }%
}%

~\,\raisebox{10.6pt}[0pt][0pt]{\box\graph}\;~,
\expandafter\ifx\csname graph\endcsname\relax
   \csname newbox\expandafter\endcsname\csname graph\endcsname
\fi
\ifx\graphtemp\undefined
  \csname newdimen\endcsname\graphtemp
\fi
\expandafter\setbox\csname graph\endcsname
 =\vtop{\vskip 0pt\hbox{%
    \special{pn 8}%
    \special{ar 39 197 39 39 0 6.28319}%
    \special{ar 433 197 39 39 0 6.28319}%
    \special{sh 1.000}%
    \special{pn 1}%
    \special{pa 315 172}%
    \special{pa 394 197}%
    \special{pa 315 222}%
    \special{pa 315 172}%
    \special{fp}%
    \special{pn 8}%
    \special{pa 79 197}%
    \special{pa 315 197}%
    \special{fp}%
    \graphtemp=\baselineskip
    \multiply\graphtemp by -1
    \divide\graphtemp by 2
    \advance\graphtemp by .5ex
    \advance\graphtemp by 0.197in
    \rlap{\kern 0.236in\lower\graphtemp\hbox to 0pt{\hss $\tau$\hss}}%
    \special{ar 827 197 39 39 0 6.28319}%
    \special{sh 1.000}%
    \special{pn 1}%
    \special{pa 852 79}%
    \special{pa 827 157}%
    \special{pa 802 79}%
    \special{pa 852 79}%
    \special{fp}%
    \special{pn 8}%
    \special{pa 827 0}%
    \special{pa 827 79}%
    \special{fp}%
    \special{sh 1.000}%
    \special{pn 1}%
    \special{pa 709 172}%
    \special{pa 787 197}%
    \special{pa 709 222}%
    \special{pa 709 172}%
    \special{fp}%
    \special{pn 8}%
    \special{pa 472 197}%
    \special{pa 709 197}%
    \special{fp}%
    \graphtemp=\baselineskip
    \multiply\graphtemp by -1
    \divide\graphtemp by 2
    \advance\graphtemp by .5ex
    \advance\graphtemp by 0.197in
    \rlap{\kern 0.630in\lower\graphtemp\hbox to 0pt{\hss $c$\hss}}%
    \hbox{\vrule depth0.236in width0pt height 0pt}%
    \kern 0.866in
  }%
}%

~\,\raisebox{10.6pt}[0pt][0pt]{\box\graph}\;~\right\}$.
\end{trivlist}

\begin{example}{SOS graph model 2}
A TSS with a rule $c\ar{a}x$ (without premises)
does not induce a process graph semantics.
Namely, no matter which set $\IG^*$ one takes, there is always a larger set $\IG'$ and a process
graph $G \mathbin\in \IG'\setminus\IG^*\!$. Thus, the process graph $\den{c}_P^{\IGs'}$ has the
transition $c \ar{a} G$ but $\den{c}_P^{\IGs^*}$ does not. Hence $\IG^*$ is not adequate.
Consequently, the TSS does not induce a process graph semantics.
\end{example}

\begin{example}{SOS graph model 1}
Let $P$ be the TSS with constants $c$ and $0$ and as only rule
$$\frac{x\ar{a}y~~z\ar{b}y}{c\ar{a}0}\;.$$
Then $\den{c}_P^{\IGs^*}(\rho)$ is independent of $\rho$, since $c$ is a closed term.
In any adequate choice of $\IG^*$ there is a graph in which an
$a$-transition and a $b$-transition end in a common state, and
using such a choice one finds that \plat{\den{c}_P^{\IGs^*}} has the transition $c\ar{a}0$.
However, when taking $\IG^*$ to be a set of trees,
\plat{\den{c}_P^{\IGs^*}} has no transitions.

$P$ induces a process graph semantics according to \df{strongly interpreting},
but not according to \df{simply interpreting}.
\end{example}
If we stick with \df{simply interpreting}, $\den{t}_P=\den{t}_P^\emptyset$ for closed terms $t$.

\begin{definition}{pure}
The \phrase{rule-bound variables} of a transition rule \plat{\frac{H}{t\ar{\scriptscriptstyle a}t'}} form the smallest set $B$ such
that
\begin{itemize}
\item $\var(t) \subseteq B$, and
\item if $u\ar{b}u'$ is a premise in $H$ and $\var(u)\subseteq B$ then $\var(u')\subseteq B$.
\end{itemize}
A TSS is called \phrase{pure} if all variables occurring free in one of its rules are rule-bound in that rule.
\end{definition}
This concept of a pure TSS generalises the one from \cite{GrV92}, and coincides with it for TSSs in
the tyft/tyxt format studied in \cite{GrV92}. The TSSs of all common process algebras are pure.
So is the TSS of \ex{context dependence}, but
the ones of Examples~\ref{ex:SOS graph model 2} and~\ref{ex:SOS graph model 1} are not. 

\begin{proposition}{pure}
Any pure TSS manifestly induces a process graph semantics.
\end{proposition}
\begin{proof}
For a given term $t\in\IT^r(\Sigma)$ and valuation $\rho:\Var\rightarrow\IG(A)$,
let $\IG^*_0$ be the smallest set of process graphs that is transition-closed and supports $(t,\rho)$.
Then for any $\IG^*\supseteq \IG^*_0$ and any transition $t[\rho] \ar{a} u$ provable from $P+\IG^*$,
any term $v$ occurring in a proof of this transition, including the target $u$,
has the form $t'[\rho]$ with $t'\in\IT^r(\Sigma)$, such that $\rho(x)\in\IG^*_0$ for any $x\in\var(t)$.
This follows by a fairly straightforward induction on the size of proofs, with a nested induction on the
derivation of rule-boundedness. As a consequence, the process graph $\den{t[\rho]}_{P+\IGs^*}^\emptyset$
does not depend in any way on $\IG^*\setminus \IG^*_0$.
\end{proof}

\paragraph{Remark} One may wonder whether the above treatment can be simplified by skipping, in
\df{interpreting}, ``the reachable part of''. The answer is negative, for in that case 
\plat{\den{t}_P^{\IGs^*}(\rho)} would never be independent of the choice of $\IG^*\!$,
because all $G \in \IG^*$ would occur as (unreachable) states in the process graph \plat{\den{t}_P^{\IGs^*}(\rho)}.

\paragraph{Summary} In this section, terms in a TSS are interpreted in the domain of process graphs as follows.
Let $P=(\Sigma,A,R)$ be a TSS and $t\in\IT^r(\Sigma)$ a term. The meaning $\den{t}_P:\IG(A)^\Var\!\rightarrow\IG(A)$
of $t$ is given by \plat{\den{t}_P(\rho) := \den{t[\rho]}_{P+\IGs^*}^\emptyset}, with $\IG^*$ adequate for $\den{t}_P(\rho)$.
Here $\IG^*$ is transition-closed and supports $(t,\rho)$; it is \phrase{adequate} if further
increasing this set does not alter the definition of $\den{t}_P(\rho)$. If no adequate $\IG^*$
exists, $\den{t}_P(\rho)$ remains undefined, and $P$ does not induce a process graph semantics.
If $P$ is pure, any transition-closed set $\IG^*\subseteq \IG(A)$ supporting $(t,\rho)$ is adequate.

\section{Lifting Semantic Equivalences to Open Terms}\label{sec:lifting}

The following definition shows how any equivalence relation $\sim$ defined on a domain $\ID$
in which a language $\fL$ is interpreted, lifts to the open terms of $\fL$.

\begin{definition}{lifting}({\em Lifting equivalences to open terms}).
For a language $\fL$, given as an annotated signature $\Sigma$
and an interpretation $\den{\_}:\IT^r(\Sigma) \rightarrow  (\ID^\Var\rightarrow\ID)$,
and an equivalence relation $\sim$ on $\ID$, write $t \sim_\fL u$ for $t,u\in\IT^r(\Sigma)$
iff $\den{t}(\rho)\sim\den{u}(\rho)$ for all valuations $\rho:\Var\rightarrow\ID$.
\end{definition}
This definition can be applied to any language $\fL$ given by a TSS $P=(\Sigma,A,R)$.
In this case $\sim$ must be defined on $\IG(A)$. Write $\sim^{\textbf{\textsf{pg}}}_P$ for $\sim_\fL$ as defined above
when taking as interpretation the process graph semantics
$\den{\_}:\IT^r(\Sigma) \rightarrow  (\IG(A)^\Var\rightarrow\IG(A))$ of $\fL$.

An equivalence $\sim$ on $\IG(A)$ also lifts to the closed terms $\T^r(\Sigma)$ of $\fL$.
Namely, let $(\T^r(\Sigma),A,\rightarrow)$ be the LTS specified by $P$ as defined in \Sec{TSS}.
Then $(\T^r(\Sigma),A,\rightarrow,p)\in \IG(A)$ is a process graph for any $p\in\T^r(\Sigma)$.
Now write $p \sim q$, for $p,q\in\T^r(\Sigma)$, whenever $(\T^r(\Sigma),A,\rightarrow,p) \sim (\T^r(\Sigma),A,\rightarrow,p)$.

Using this, \df{lifting} can also be instantiated by taking  as interpretation
the closed-term semantics $\den{\_}:\IT^r(\Sigma) \rightarrow (\T^r(\Sigma)^\Var\rightarrow\T^r(\Sigma))$
of $\fL$, as defined in \Sec{closed-term}. Write $\sim^ {\textbf{\textsf{ci}}}_P$ for $\sim_\fL$ defined thusly.
So
\begin{center}
  $t \sim^ {\textbf{\textsf{ci}}}_P u$ iff $t[\sigma] \sim u[\sigma]$ for any closed substitution $\sigma$,
\end{center}
i.e., two open terms are related by $\sim^ {\textbf{\textsf{ci}}}_P$ if all of their \emph{closed instantiations} are
related by $\sim$.

Having lifted semantic equivalences $\sim$ from process graphs to open terms in two ways,
one wonders how the resulting equivalences compare. Instantiating $\sim$ with strong bisimilarity,
$\bis{}$, \ex{context dependence} shows that $f(x) \bis{}^ {\textbf{\textsf{ci}}}_P id(x)$ yet $f(x) \mathrel{\nobis{}^ {\textbf{\textsf{pg}}}_P} id(x)$.
For the other direction consider the TSS with constants $0$, $c$ and $d$, alphabet $\{a,b\}$, and the
rules\vspace{-2ex} $$d\ar{a}0 \qquad \frac{x\ar{b} y}{c\ar{a} 0}\;.$$
Under the closed-term interpretation, no $b$-transition from any term can be derived, so
$0 \bis{}^ {\textbf{\textsf{ci}}} c \mathrel{\nobis{}^ {\textbf{\textsf{ci}}}} d$.
Yet under the process graph interpretation, since there exists some graph that can do a
$b$-transition, one has $c\ar{a}0$, and obtains $0 \mathrel{\nobis{}^ {\textbf{\textsf{pg}}}} c \bis{}^ {\textbf{\textsf{pg}}} d$.
So in general $\bis{}^ {\textbf{\textsf{ci}}}$ and $\bis{}^ {\textbf{\textsf{pg}}}$ are incomparable.

The above TSS is not pure; the variable $x$ is not rule-bound. For pure TSSs no such example exists.
\begin{theorem}{finer}
Let $P=(\Sigma,A,R)$ be a pure TSS and $\approx$ an equivalence on $\IG(A)$ that relates each process
graph with its reachable part. Moreover, let ${\sim} \subseteq {\approx}$ be a possibly finer, or more
discriminating, equivalence that satisfies requirements (\ref{variable interpretation})--(\ref{alpha-conversion1})
of \Sec{requirements}. Then $t\approx^ {\textbf{\textsf{pg}}}_P u$ implies $t\approx^ {\textbf{\textsf{ci}}}_P u$.
\end{theorem}

\begin{proof}
  Suppose $t\approx^ {\textbf{\textsf{pg}}}_P u$, and let $\sigma:\Var\rightarrow\T^r(\Sigma)$ be a closed substitution.
  It suffices to establish that $t[\sigma] \approx u[\sigma]$.
  Let $\den{\sigma}_P: \Var \rightarrow \IG(A)$ be the valuation defined by $\den{\sigma}_P(X) := \den{\sigma(X)}_P$.
  This is the definition of $\den{\sigma}_P$ from \Sec{substitutions}, specialised to closed substitutions.
  Here $\den{q}_P$, for $q\in\T^r(\Sigma)$, is the process graph semantics of $q$ as defined in \Sec{pg semantics},
  so \plat{\den{q}_P = \den{q}^\emptyset_{P+\IGs^*}} for an adequate choice of $\IG^*$.
  Since $P$ is pure and $q$ closed, the empty set of process graphs is adequate by \pr{pure}.
  By \df{interpreting}, $\den{q}^\emptyset_P$ is the reachable part of the process graph $(\T^r(\Sigma),A,\rightarrow,p)$,
  so $\den{q}^\emptyset_P \approx (\T^r(\Sigma),A,\rightarrow,p)$.
  Since $t\approx^ {\textbf{\textsf{pg}}}_P u$, one has $\den{t}_p(\den{\sigma}_P) \approx \den{u}_p(\den{\sigma}_P)$ by \df{lifting}.
  Moreover, $\den{t}_p(\den{\sigma}_P) \sim \den{t[\sigma]}_P$ by \pr{inductive meaning}.
  Hence $$\den{t[\sigma]}^\emptyset_{P}  = \den{t[\sigma]}_P \approx \den{t}_p(\den{\sigma}_P)
  \approx \den{u}_p(\den{\sigma}_P) \approx \den{u[\sigma]}_P = \den{u[\sigma]}^\emptyset_{P}$$
  and thus $(\T^r(\Sigma),A,\rightarrow,t[\sigma]) \approx \den{t[\sigma]}^\emptyset_{P+\IGs^*}
  \approx \den{u[\sigma]}^\emptyset_{P} \approx (\T^r(\Sigma),A,\rightarrow,u[\sigma])$,
  i.e., $t[\sigma] \approx u[\sigma]$.
\end{proof}
So, under the conditions of \thm{finer}, $\approx^ {\textbf{\textsf{pg}}}_P$ is a finer, or more discriminating, equivalence than $\approx^ {\textbf{\textsf{ci}}}_P$.

\section{Congruence Properties}\label{sec:congruence}

Whether a semantic equivalence $\sim$ is a congruence (cf.\ (\ref{congruence}) in \Sec{quotient}) may depend on whether the closed-term or the
process graph semantics is chosen. The following example illustrates this for a practical process algebra.

\begin{example}{congruence}
Consider the TSS with constants $1$ and $\alpha$ for $\alpha\in Act = A\uplus\{\tau\}$ and binary
operators $+$ and $;$, denoting choice and sequencing in a process algebra,
with the following transition rules:
\[\quad \alpha \ar{\alpha} 1 \qquad
  \frac{x \ar{\alpha} x'}{x+y \ar{\alpha} x'} \qquad
  \frac{y \ar{\alpha} y'}{x+y \ar{\alpha} y'} \qquad
  \frac{x \ar{\alpha} x'}{x;y \ar{\alpha} x';y} \qquad
  \frac{x \nar{\alpha} ~\mbox{for all}~ \alpha\in Act \quad y \ar{\beta} y'}{x;y \ar{\beta} y'} \qquad
\]
The process 1, like 0 in CCS, has no outgoing transitions, meaning that it performs no actions.
The sequencing operator performs all actions its first argument can do, until its first argument can
perform no further actions; then it continues with its second argument. I employ no recursion here.

As equivalence relation $\sim$ I take weak bisimulation equivalence, $\bis{w}$, as defined in \cite{Mi90ccs,vG00}.
For the term $t$ from (\ref{congruence}) take $x;b$. Let $\rho$ and $\nu$ be valuations with
\begin{center}
\expandafter\ifx\csname graph\endcsname\relax
   \csname newbox\expandafter\endcsname\csname graph\endcsname
\fi
\ifx\graphtemp\undefined
  \csname newdimen\endcsname\graphtemp
\fi
\expandafter\setbox\csname graph\endcsname
 =\vtop{\vskip 0pt\hbox{%
    \special{pn 8}%
    \special{ar 390 65 65 65 0 6.28319}%
    \special{sh 1.000}%
    \special{pn 1}%
    \special{pa 225 40}%
    \special{pa 325 65}%
    \special{pa 225 90}%
    \special{pa 225 40}%
    \special{fp}%
    \special{pn 8}%
    \special{pa 0 65}%
    \special{pa 225 65}%
    \special{fp}%
    \special{ar 1039 65 65 65 0 6.28319}%
    \special{sh 1.000}%
    \special{pn 1}%
    \special{pa 874 40}%
    \special{pa 974 65}%
    \special{pa 874 90}%
    \special{pa 874 40}%
    \special{fp}%
    \special{pn 8}%
    \special{pa 455 65}%
    \special{pa 874 65}%
    \special{fp}%
    \graphtemp=\baselineskip
    \multiply\graphtemp by -1
    \divide\graphtemp by 2
    \advance\graphtemp by .5ex
    \advance\graphtemp by 0.065in
    \rlap{\kern 0.714in\lower\graphtemp\hbox to 0pt{\hss $a$\hss}}%
    \hbox{\vrule depth0.130in width0pt height 0pt}%
    \kern 1.104in
  }%
}%

$\rho(x) = ~\,\raisebox{13.6pt}[0pt][0pt]{\box\graph}\;~$ \qquad and \qquad
\expandafter\ifx\csname graph\endcsname\relax
   \csname newbox\expandafter\endcsname\csname graph\endcsname
\fi
\ifx\graphtemp\undefined
  \csname newdimen\endcsname\graphtemp
\fi
\expandafter\setbox\csname graph\endcsname
 =\vtop{\vskip 0pt\hbox{%
    \special{pn 8}%
    \special{ar 390 91 65 65 0 6.28319}%
    \special{sh 1.000}%
    \special{pn 1}%
    \special{pa 225 66}%
    \special{pa 325 91}%
    \special{pa 225 116}%
    \special{pa 225 66}%
    \special{fp}%
    \special{pn 8}%
    \special{pa 0 91}%
    \special{pa 225 91}%
    \special{fp}%
    \special{ar 1039 91 65 65 0 6.28319}%
    \special{sh 1.000}%
    \special{pn 1}%
    \special{pa 874 66}%
    \special{pa 974 91}%
    \special{pa 874 116}%
    \special{pa 874 66}%
    \special{fp}%
    \special{pn 8}%
    \special{pa 455 91}%
    \special{pa 874 91}%
    \special{fp}%
    \graphtemp=\baselineskip
    \multiply\graphtemp by -1
    \divide\graphtemp by 2
    \advance\graphtemp by .5ex
    \advance\graphtemp by 0.091in
    \rlap{\kern 0.714in\lower\graphtemp\hbox to 0pt{\hss $a$\hss}}%
    \special{sh 1.000}%
    \special{pn 1}%
    \special{pa 1171 193}%
    \special{pa 1085 137}%
    \special{pa 1188 146}%
    \special{pa 1171 193}%
    \special{fp}%
    \special{pn 8}%
    \special{pa 1085 45}%
    \special{pa 1234 -5}%
    \special{pa 1429 -5}%
    \special{pa 1429 188}%
    \special{pa 1234 188}%
    \special{pa 1094 140}%
    \special{sp}%
    \graphtemp=.5ex
    \advance\graphtemp by 0.091in
    \rlap{\kern 1.519in\lower\graphtemp\hbox to 0pt{\hss $\tau$\hss}}%
    \hbox{\vrule depth0.182in width0pt height 0pt}%
    \kern 1.519in
  }%
}%

$\nu(x) = ~\,\raisebox{13.6pt}[0pt][0pt]{\box\graph}\;~$ \qquad.
\end{center}
Then $\rho(x) \bis{w} \nu(x)$, so that I may assume $\rho \bis{w} \nu$.
Now the term $t$ performs the sequential composition of the process filled in for $x$ with the process
doing a single $b$ action. One has $\den{x;b}(\rho)\nobis_w\den{x;b}(\nu)$ because only the first of
these processes can ever perform the $b$. Thus, when using the process graph semantics of this TSS,
$\bis{w}$ fails to be a congruence for the language specified.

However, when taking the closed-term semantics, all processes that my be filled in for $x$ are terms
in the given language and thus must terminate after performing finitely many transitions.
In this setting $\bis{w}$ is actually a congruence. However, it stops being a congruence when
recursion is added to the language.
\end{example}

\section{Relating Sanity Requirements for the Two Semantics of TSSs}\label{sec:relating}

Given an equivalence relation $\sim$ on $\IG(A)$, let $\sim^{\textbf{\textsf{cl}}}_P$ be the equivalence relation
on the set $\T^r(\Sigma)$ of closed terms of a TSS $P=(\Sigma,A,R)$ defined by
$p \sim^{\textbf{\textsf{cl}}}_P q$ iff $\den{p}_P^\emptyset \sim \den{q}_P^\emptyset$ (cf.\ \df{interpreting}).
In case $\sim$ relates each process graph with its reachable part, the equivalence $\sim^{\textbf{\textsf{cl}}}_P$
coincides with $\sim^ {\textbf{\textsf{ci}}}_P$, as defined in \Sec{lifting}.

\begin{observation}{pure}
If $P$ is pure and $p,q\in\T^r(\sigma)$, then $p \sim^{\textbf{\textsf{cl}}}_P q$ iff $p \sim^ {\textbf{\textsf{pg}}}_P q$.
\end{observation}

\begin{theorem}{sanity}
Let $P$ be a TSS that induces a process graph semantics $\den{\_}_P$
and let $\sim$ be an equivalence relation on $\IG(A)$.
Then $\den{\_}_P$ satisfies the sanity requirements
(\ref{comp-operators})--(\ref{solutions}) of Section~\ref{sec:semantics} up to $\sim$
if the closed-term semantics of $P+\IG^*$ satisfies these requirements up to $\sim^{\textbf{\textsf{cl}}}_{P+\IGs^*}$
for any choice of $\IG^*$.
\end{theorem}

\begin{proof}
Let $\rho,\nu:\Var\rightarrow\IG(A)$, $(f,n)\in\Sigma$ and $t_i,u_i\in\IT^r(\Sigma)$,
such that $\den{t_i}_P(\rho)\sim\den{u_i}_P(\nu)$ for $i=1,...,n$.
Let the set $\IG^*\subseteq \IG(A)$ be adequate for the definition of $\den{t_i}_P(\rho)$ and
$\den{u_i}_P(\nu)$ for $i=1,...,n$ as well as for $\den{f(t_1,...,t_n)}_P(\rho)$ and
$\den{f(u_1,...,u_n)}_P(\nu)$. Then $\den{t_i[\rho]}_{P+\IGs^*}^\emptyset \sim \den{u_i[\nu]}_{P+\IGs^*}^\emptyset$,
i.e., $t_i[\rho] \sim^{\textbf{\textsf{cl}}}_{P+\IGs^*} u_i[\nu]$, for $i=1,...,n$.
So \plat{f(t_1[\rho],...,t_n[\rho]) \sim^{\textbf{\textsf{cl}}}_{P+\IGs^*} f(u_1[\nu],...,u_n[\nu])}
by Requirement~(\ref{comp-operators-closed}) for the closed-term semantics of $P+\IG^*$; that is,
$\den{f(t_1[\rho],...,t_n[\rho])}_{P+\IGs^*}^\emptyset \sim \den{f(u_1[\nu],...,u_n[\nu])}_{P+\IGs^*}^\emptyset$,
or $\den{f(t_1,...,t_n)}_P(\rho) \sim \den{f(u_1,...,u_n)}_P(\nu)$.

Let $\rho,\!\nu\!:\Var\mathbin\rightarrow\IG(A)$ and $S,S'\!:W\mathbin\rightarrow \IT^r(\Sigma)$
with $X\mathbin\in W \mathbin\subseteq \Var$,
such that $\den{S_Y}_P(\rho\backslash\! W) \sim \den{S'_Y}_P(\nu\backslash\! W)$ for $Y\mathbin\in W$.
The latter means that $\den{S_Y}_P(\rho\backslash W)(\xi) \sim \den{S'_Y}_P(\nu\backslash W)(\xi)$
for all $Y\in W$ and $\xi:W\rightarrow\IG(A)$.
Let $\IG^*\subseteq \IG(A)$ be adequate for $\den{S_Y}_P(\xi \cup \rho\backslash W)$
and $\den{S'_Y}_P(\xi \cup \nu\backslash W)$ for all $Y\in W$ and $\xi:W\rightarrow\IG(A)$,\linebreak
as well as for $\den{\rec{X|S}}_P(\rho)$ and $\den{\rec{X|S'}}_P(\nu)$.
Then $S_Y[\rho\backslash W][\xi] \sim^{\textbf{\textsf{cl}}}_{P+\IGs^*} S'_Y[\nu\backslash W][\xi]$
for all $Y\in W$ and $\xi:W\mathbin\rightarrow\IG(A)$.
So $\rec{X|S[\rho\backslash W]} \sim^{\textbf{\textsf{cl}}}_{P+\IGs^*} \rec{X|S'[\nu\backslash W]}$
by Requirement~(\ref{comp-recursion-closed}) for the
\hyperlink{closed-term}{closed-term semantics} of $P{+}\IG^*$, and
$\den{\rec{X|S}}_P(\rho) = \den{\rec{X|S}[\rho]}^\emptyset_{P+\IGs^*} =
 \den{\rec{X|S[\rho\backslash W]}}^\emptyset_{P+\IGs^*} \sim
 \den{\rec{X|S'[\nu\backslash W]}}^\emptyset_{P+\IGs^*} =
 \den{\rec{X|S'}}_P(\nu)$.

Let $S:V_S\rightarrow\IT^r(\Sigma)$ with $X\in V_S\subseteq \Var$, and let $\gamma:V_S\rightarrow \Var$ be
an injective substitution such that the range of $\gamma$ contains no variables occurring free in
$\rec{S_Y|S}$ for some $Y\in V_S$. Take $\rho:\Var\rightarrow\IG(A)$.
Let $\IG^*$ be adequate for $\den{\rec{X|S}}_P(\rho)$
and $\den{\rec{\gamma(X)|S[\gamma]}}_P(\rho)$.
By Requirement~(\ref{alpha-conversion1-closed}) for the closed-term semantics of $P+\IG^*$ one has
$\rec{\gamma(X)|S[\gamma]}[\rho] \sim^{\textbf{\textsf{cl}}}_{P+\IGs^*} \rec{X|S}[\rho]$,
using that $\rec{X|S}[\rho] = \rec{X|S[\rho\backslash V_S]}$ and $\rec{\gamma(X)|S[\gamma]}[\rho]
= \rec{\gamma(X)|S[\gamma][\rho\backslash\gamma(V_S)]} = \rec{\gamma(X)|S[\rho\backslash V_S][\gamma]}$.
So $\den{\rec{\gamma(X)|S[\gamma]}}_{P}(\rho) \sim \den{\rec{X|S}}_{P}(\rho)$.

Let $S:V_S\rightarrow\IT^r(\Sigma)$ with $X\in V_S\subseteq \Var$.
Take $\rho:\Var\rightarrow\IG(A)$.
Let $\IG^*$ be adequate for $\den{\rec{X|S}}_P(\rho)$ and $\den{S_X|S}_P(\rho)$. Then
$\rec{X|S[\rho\backslash V_S]} \sim^{\textbf{\textsf{cl}}}_{P+\IGs^*} \rec{S_X[\rho\backslash V_S]|S[\rho\backslash V_S]}$
by Requirement~(\ref{solutions-closed}) for the closed-term semantics of $P+\IG^*\!\!$.
Hence
$\den{\!\rec{X|S}\!}_P(\rho) \mathbin= \den{\rec{X|S}[\rho]}_{\scriptscriptstyle P+\IGss^*}^\emptyset \mathbin=
\den{\rec{X|S[\rho\!\backslash\! V_S]}}_{\scriptscriptstyle P+\IGss^*}^\emptyset \mathbin\sim 
\den{\rec{S_X[\rho\!\backslash\! V_S]|S[\rho\!\backslash\! V_S]}}_{\scriptscriptstyle P+\IGss^*}^\emptyset\linebreak[3] =
\den{\rec{S_X|S}[\rho]}_{\scriptscriptstyle P+\IGss^*}^\emptyset =
\den{S_X|S}_P(\rho)$.
\end{proof}

\thm{sanity} does not extend to sanity requirement (\ref{variable interpretation}).
In fact, this requirement always holds for the closed-term interpretation of a TSS,
yet it does not always hold for the process graph interpretation:

\begin{example}{variable interpretation}
Let $P$ be a TSS with the single rule $\frac{x\ar{\tau}y,~~y\ar{\tau} z}{x\ar{\tau} z}$.
Then the process graph semantics of $P$ fails to satisfy sanity requirement
(\ref{variable interpretation}) up to $\bis{}$. Namely, if $\rho(x)$ is a graph
\expandafter\ifx\csname graph\endcsname\relax
   \csname newbox\expandafter\endcsname\csname graph\endcsname
\fi
\ifx\graphtemp\undefined
  \csname newdimen\endcsname\graphtemp
\fi
\expandafter\setbox\csname graph\endcsname
 =\vtop{\vskip 0pt\hbox{%
    \special{pn 8}%
    \special{ar 236 39 39 39 0 6.28319}%
    \special{sh 1.000}%
    \special{pn 1}%
    \special{pa 118 14}%
    \special{pa 197 39}%
    \special{pa 118 64}%
    \special{pa 118 14}%
    \special{fp}%
    \special{pn 8}%
    \special{pa 0 39}%
    \special{pa 118 39}%
    \special{fp}%
    \special{ar 630 39 39 39 0 6.28319}%
    \special{sh 1.000}%
    \special{pn 1}%
    \special{pa 512 14}%
    \special{pa 591 39}%
    \special{pa 512 64}%
    \special{pa 512 14}%
    \special{fp}%
    \special{pn 8}%
    \special{pa 276 39}%
    \special{pa 512 39}%
    \special{fp}%
    \graphtemp=\baselineskip
    \multiply\graphtemp by -1
    \divide\graphtemp by 2
    \advance\graphtemp by .5ex
    \advance\graphtemp by 0.039in
    \rlap{\kern 0.433in\lower\graphtemp\hbox to 0pt{\hss $\tau$\hss}}%
    \special{ar 1024 39 39 39 0 6.28319}%
    \special{sh 1.000}%
    \special{pn 1}%
    \special{pa 906 14}%
    \special{pa 984 39}%
    \special{pa 906 64}%
    \special{pa 906 14}%
    \special{fp}%
    \special{pn 8}%
    \special{pa 669 39}%
    \special{pa 906 39}%
    \special{fp}%
    \graphtemp=\baselineskip
    \multiply\graphtemp by -1
    \divide\graphtemp by 2
    \advance\graphtemp by .5ex
    \advance\graphtemp by 0.039in
    \rlap{\kern 0.827in\lower\graphtemp\hbox to 0pt{\hss $\tau$\hss}}%
    \hbox{\vrule depth0.079in width0pt height 0pt}%
    \kern 1.063in
  }%
}%

\;\raisebox{13.6pt}[0pt][0pt]{\box\graph}\;,
then $\den{x}_P(\rho)$ is the graph
\expandafter\ifx\csname graph\endcsname\relax
   \csname newbox\expandafter\endcsname\csname graph\endcsname
\fi
\ifx\graphtemp\undefined
  \csname newdimen\endcsname\graphtemp
\fi
\expandafter\setbox\csname graph\endcsname
 =\vtop{\vskip 0pt\hbox{%
    \special{pn 8}%
    \special{ar 236 59 39 39 0 6.28319}%
    \special{sh 1.000}%
    \special{pn 1}%
    \special{pa 118 34}%
    \special{pa 197 59}%
    \special{pa 118 84}%
    \special{pa 118 34}%
    \special{fp}%
    \special{pn 8}%
    \special{pa 0 59}%
    \special{pa 118 59}%
    \special{fp}%
    \graphtemp=.5ex
    \advance\graphtemp by 0.000in
    \rlap{\kern 0.433in\lower\graphtemp\hbox to 0pt{\hss $\tau$\hss}}%
    \special{ar 630 59 39 39 0 6.28319}%
    \graphtemp=.5ex
    \advance\graphtemp by 0.000in
    \rlap{\kern 0.827in\lower\graphtemp\hbox to 0pt{\hss $\tau$\hss}}%
    \special{sh 1.000}%
    \special{pn 1}%
    \special{pa 512 34}%
    \special{pa 591 59}%
    \special{pa 512 84}%
    \special{pa 512 34}%
    \special{fp}%
    \special{pn 8}%
    \special{pa 276 59}%
    \special{pa 512 59}%
    \special{fp}%
    \special{ar 1024 59 39 39 0 6.28319}%
    \special{sh 1.000}%
    \special{pn 1}%
    \special{pa 906 34}%
    \special{pa 984 59}%
    \special{pa 906 84}%
    \special{pa 906 34}%
    \special{fp}%
    \special{pn 8}%
    \special{pa 669 59}%
    \special{pa 906 59}%
    \special{fp}%
    \special{sh 1.000}%
    \special{pn 1}%
    \special{pa 919 117}%
    \special{pa 996 87}%
    \special{pa 951 156}%
    \special{pa 919 117}%
    \special{fp}%
    \special{pn 8}%
    \special{ar 630 -298 531 531 0.959715 2.330025}%
    \graphtemp=.5ex
    \advance\graphtemp by 0.177in
    \rlap{\kern 0.630in\lower\graphtemp\hbox to 0pt{\hss $\tau$\hss}}%
    \hbox{\vrule depth0.233in width0pt height 0pt}%
    \kern 1.063in
  }%
}%

\,\raisebox{10.6pt}[0pt][0pt]{\box\graph}\;,
and the two graphs are not (strongly) bisimulation equivalent.
\end{example}
Nevertheless, requirement (\ref{variable interpretation}) holds for almost all process algebras
found in the literature:

\begin{proposition}{variable interpretation}
Let $P$ be a TSS that has no rule with a variable as the left-hand side of the conclusion.
The process-graph interpretation of $P$ always
satisfies requirement (\ref{variable interpretation}) of Section~\ref{sec:semantics} up to $\bis{}$.
\end{proposition}

\begin{proof}
For $x\in \Var$ and $\rho:\Var\rightarrow\IG(A)$
let $\IG^* \subseteq \IG(A)$ be adequate for $\den{x}_P(\rho)$ and
let $\rho(x)\in\IG^*$ be the graph $g=(S,\rightarrow,I)$.
For each state $s\in S$ there is a graph $g_s:=(S,\rightarrow,s)$ in $\IG^*$.
Now $\den{x}_P(\rho)$ is the reachable part of the graph
$(G,\rightarrow_G,g)$, where $G=\{g_s\mid s\in S\}$ and
$\mathord{\rightarrow_G} = \{(g_s,a,g_t) \mid (s,a,t)\in \mathord{\rightarrow}\}$.
The relation $\R$ given by $g_s \R s$ for all states $s\in S$
reachable from $I$ clearly is a bisimulation. Therefore
$\den{x}_P(\rho) \bis{} \rho(x)$.
\end{proof}

\section{\texorpdfstring{$\!\!\!$}{}%
         Preservation of Relative Expressiveness under Conservative Extensions}\label{sec:conservative}

\begin{definition}{sum of TSSs}
If $P=(\Sigma_P,A,R_P)$ and $Q=(\Sigma_Q,A,R_Q)$ are TSSs
with $\Sigma_P$ and $\Sigma_Q$ disjoint
then $P+Q$ denotes the TSS $(\Sigma_P\cup\Sigma_Q,A,R_P\cup R_Q)$.
\end{definition}
Let $P_0$ be the TSS of \ex{context dependence}, but without the operator $f$.
Let $P_f$ be the part of the TSS of \ex{context dependence} that only contains the operator $f$,
so that the entire TSS of \ex{context dependence} is $P_0+P_f$.
This TSS does not feature recursion.

A \emph{translation} between two languages with signatures $\Sigma$ and $\Sigma'$ is a mapping
$\fT:\IT^r(\Sigma) \rightarrow \IT^r(\Sigma')$. Consider the translation $\fT_{id}$ from the
language specified by $P_0$ to the language specified by $P_0+P_f$, given by $\fT_{id}(t)=t$ for
all $t\in\Sigma_{P_0}$. Also consider the translation $\fT_{\it op}$ in the opposite direction, given by
$\fT_{\it op}(a.t):= a.\fT_{\it op}(t)$, $\fT_{\it op}(id(t)):= id(\fT_{\it op}(t))$
and $\fT_{\it op}(f(t)):= id(\fT_{\it op}(t))$ for all $t\in\IT(\Sigma_{P_o+P_f})$.

In e.g.\ \cite{vG18b} a concept of expressiveness of specification languages is studied
such that language $\fL'$ is at least as expressive as language $\fL$ up to a semantic equivalence
relation $\sim$ iff there exist a translation from $\fL$ into $\fL'$ that is valid up to $\sim$.
It is not important here to state the precise definition of validity; it suffices to point out that
$\fT_{\it op}$ is valid up to $\bis{}$ iff one has $f(x) \bis{} id(x)$. Thus, applying \df{lifting}, it is
valid when employing the closed-term semantics, but not when employing the process graph semantics.\pagebreak[3]
The transition $\fT_{id}$ on the other hand is valid up to $\bis{}$ regardless which
of the two interpretations one picks. So under the closed-term interpretation the two languages are
equally expressive, whereas under the process graph semantics the language given by $P_0+P_f$
is more expressive than the one given by $P_0$.

For TSSs $P$ and $Q$, write $P \preceq Q$ if the language specified by $Q$ is at least as expressive
as the one specified by $P$. An intuitively plausible theorem is that\vspace{-2pt}
\begin{equation}\label{cons}
\mbox{$P_1 \preceq P_2$ \quad implies \quad $P_1+Q \preceq P_2+Q$,}
\end{equation}
at least under some mild conditions on the TSSs $P_1$, $P_2$ and $Q$, for instance that they are pure and fit the
\emph{tyft} format defined in \cite{GrV92}.\footnote{These mild conditions should ensure that $P_i+Q$ is a
\emph{conservative extension} of $P_i$, for $i=1,2$, as defined in \cite{GrV92}.}
This theorem fails when employing the closed-term semantics of TSSs: take $P_1$ to be $P_0+P_f$,
$P_2$ to be $P_0$, with $\fT_{\it op}$ being the witness for $P_1 \preceq P_2$,
and $Q$ to be the TSS with as single operator $\tau.\_$ and as only transition
rule $\tau.x \ar{\tau} x$. For the operator $f$ in the TSS $P_0+P_f+Q$ drops $\tau$-transitions, and
has no counterpart in the TSS $P_0+Q$.

This problem is fixed when employing the process graph semantics.
Once the omitted definition of validity is supplied \cite{vG18b},
the proof of (\ref{cons}) is entirely straightforward.

\section{Related Work}\label{sec:related}

Dissatisfaction with the traditional closed-term interpretation of TSSs occurred earlier in
\cite{LL91,LV96,GM00,Re00} and \cite{BBB07}. However, rather than adapting the interpretation of
TSSs, as in the present paper, these papers abandon the notion of a TSS in favour of different
frameworks of system specification that are arguably more suitable for giving meaning to open terms.
Larsen and Liu \cite{LL91} use \emph{context systems}. The CCS transition rule\vspace{-9pt}
\[
  \frac{x \ar{a} x' \quad y \ar{\bar a} y'}{x|y \ar{\tau} x'|y'}
\]
for instance takes in a context system the shape
\plat{x|y \begin{array}{c}\scriptstyle \tau\\[-8pt]\longrightarrow\\[-9.5pt]\scriptstyle a,\bar a\;\mbox{}\end{array} x'|y'},
or rather, suppressing the redundant variable names,
\plat{| \begin{array}{c}\scriptstyle \tau\\[-8pt]\longrightarrow\\[-9.5pt]\scriptstyle a,\bar a\;\mbox{}\end{array} |}.
It says that the operator $|$ can perform a $\tau$-transition, provided its first argument does an
$a$-transition, and its second argument an $\bar a$. The context systems of \cite{LL91} form the
counterpart of TSSs in the De Simone format \cite{dS85}. The model is generalised by Lynch \&
Vaandrager \cite{LV96} to \emph{action transducers}, by Gadducci \& U. Montanari
\cite{GM00} to the \emph{tile model}, and by Rensink \cite{Re00} to \emph{conditional transition systems}.
The latter two proposals are further generalised to \emph{symbolic transition systems} by
Baldan, Bracciali \& Bruni \cite{BBB07}.

One method to relate these models with TSSs under the closed-term and process graph interpretations
is through notions of strong bisimilarity on open terms. This is a central theme
in \cite{Re00}. The most natural notion of bisimulation on the above models is
\emph{bisimulation under formal hypothesis}, $\bis{}^{\textbf{\textsf{fh}}}$. That name stems from De
Simone \cite{dS85}, who defined the same concept in terms of TSSs. On the context systems sketched
above it requires the usual transfer property for bisimulations for doubly labelled transitions
such as \plat{\begin{array}{c}\scriptstyle \tau\\[-8pt]\longrightarrow\\[-9.5pt]\scriptstyle a,\bar a\;\mbox{}\end{array}}.
On TSSs, a bisimulation under formal hypothesis essentially is a symmetric relation $\R$ on open
terms such that
\begin{center}
if $t \R u$ and $P\vdash\displaystyle\frac{\{x_i \ar{a_i} y_i \mid i\in I\}}{t\ar{a}t'}$
then $P\vdash\displaystyle \frac{\{x_i \ar{a_i} y_i \mid i\in I\}}{u\ar{a}u'}$
for an $u'\in \IT^r(\Sigma_Q)$ with $t'\R u'$.
\end{center}
Rensink \cite{Re00} shows that $\bis{}^{\textbf{\textsf{fh}}}$ is strictly finer than
$\bis{}^{\textbf{\textsf{ci}}}$.

\begin{example}{fh}
  Let $P$ be the TSS with inaction $0$, action prefix, choice and intersection, specified by the following rules:\vspace{-3ex}
  \[a.x \ar{a} x \qquad
    \frac{x \ar{a} x'}{x+y \ar{a} x'} \qquad
    \frac{y \ar{a} y'}{x+y \ar{a} y'} \qquad
    \frac{x \ar{a} x' \quad y \ar{a} y'}{x\cap y \ar{a} x' \cap y'}
  \pagebreak[4]\]
  where $a$ ranges over a set of actions $A$.
  Then $b.0 + b.a.0 + b.(x \cap a.0) \bis{}^{\textbf{\textsf{ci}}} b.0 + b.a.0$, for no matter what one fills in for
  $x$, the process $x \cap a.0$  either cannot perform any transitions, or it can only do a
  single $a$. So the term $b.(x \cap a.0)$ behaves either like $b.0$ or like $b.a.0$.
  On the other hand, $b.0 + b.a.0 + b.(x \cap a.0) \nobis^{\,\textbf{\textsf{fh}}} b.0 + b.a.0$.
  Namely any bisimulation under formal hypothesis $\R$ relating $b.0 + b.a.0 + b.(x \cap a.0)$ with
  $b.0 + b.a.0$ would also have to relate $x \cap a.0$ with either $0$ or $a.0$.
  However, once this choice is made, substituting the wrong value for $x$ shows that $\R$ relates
  two terms that are not equivalent.
\end{example}
Rensink also defines a \emph{hypotheses-preserving bisimulation equivalence}
$\bis{}^{\textbf{\textsf{hp}}}$ on open terms, which is situated strictly between
$\bis{}^{\textbf{\textsf{fh}}}$ and $\bis{}^{\textbf{\textsf{ci}}}$.
One has $b.0 + b.a.0 + b.(x \cap a.0) \nobis^{\,\textbf{\textsf{hp}}} b.0 + b.a.0$.
His analysis can be reused to show that $\bis{}^{\textbf{\textsf{hp}}}$ is finer than 
$\bis{}^{\textbf{\textsf{pg}}}$. Note that $b.0 + b.a.0 + b.(x \cap a.0)
\bis{}^{\textbf{\textsf{pg}}} b.0 + b.a.0$, for the same reasons as in \ex{fh}.
Thus, under the conditions of \thm{finer}, we arrive at a hierarchy\vspace{-2pt}
$${\bis{}^{\textbf{\textsf{fh}}}} \quad\subseteq\quad
  {\bis{}^{\textbf{\textsf{hp}}}} \quad\subseteq\quad
  {\bis{}^{\textbf{\textsf{pg}}}} \quad\subseteq\quad
  {\bis{}^{\textbf{\textsf{ci}}}}\quad.$$

\section{Concluding Remarks}\label{sec:conclusion}

This paper proposed a process graph semantics of TSSs as an alternative to the traditional
closed-term semantics. It interprets an operator from the language as an operation on process
graphs. Unlike the closed-term semantics, this interpretation is independent of the selection of
processes that are expressible in the TSS as a whole. The intuitively plausible statement that an
expressiveness inclusion between languages is preserved under a conservative extension of source
and target language alike, fails for the closed-term semantics but holds for the proposed process
graph semantics.

I reviewed five sanity requirements on languages equipped with a semantic equivalence relation
$\sim$, and showed that four of them hold under the process semantics of a language if they hold under the
closed-term semantics. Here I end with a few observations on when these requirements hold at all.

In \cite{vG17b}, the \emph{ntyft/ntyxt format with recursion} is introduced.
It defines a wide class of TSSs, containing many known process algebras, including
CCS, CSP, ACP, {\sc Meije} and SCCS\@. It generalises the ntyft/ntyxt format of \cite{Gr93} by the
addition of recursion as a separate language construct.
The \emph{tyft/tyxt format with recursion} is the same, but not allowing negative premises.
\cite{vG17b} shows that all languages specified by a TSS in the ntyft/ntyxt format with recursion
satisfy property (\ref{congruence}) up to $\bis{}$, saying that strong bisimilarity is a congruence.
This is a stronger property than (\ref{comp-operators}) up to $\bis{}$, which thus also holds for
the ntyft/ntyxt format.
This was shown for the closed-term interpretation of TSSs. By \thm{sanity} we now also have
(\ref{comp-operators}) up to $\bis{}$ for the process graph semantics of pure TSSs in the
ntyft/ntyxt format with recursion.

The same paper establishes that (\ref{comp-recursion}) holds up to $\bis{}$ for the closed-term semantics of all
TSSs in the tyft/tyxt format with recursion, thereby generalising a result from \cite{Re00}.
It thus also holds for the process graph semantics of all pure TSSs in the tyft/tyxt format with recursion.

It is not hard to show that also requirements (\ref{alpha-conversion1}) and (\ref{solutions})
hold up to $\bis{}$ for the closed-term interpretation of TSSs in the ntyft/ntyxt format with recursion, and thus
for the process graph semantics of pure TSSs in the ntyft/ntyxt format with recursion.

Thanks to the equational nature of requirements (\ref{variable interpretation}), (\ref{alpha-conversion1}) and 
(\ref{solutions}), once they hold up to $\bis{}$, they surely hold up to any coarser equivalence.
This covers most semantic equivalences found in the literature.
The same cannot be said for requirements (\ref{comp-operators}) and (\ref{comp-recursion}).
These need to be reestablished for each semantic equivalence. There is a lot of work on congruence
formats, ensuring (\ref{comp-operators}) for a variety of semantic equivalence. See for instance
\cite{FGL17a}, and references therein. Yet, besides \cite{Re00} and \cite{vG17b} I know of no
congruence formats targeting requirement (\ref{comp-recursion}).

\paragraph{Acknowledgement}
My thanks to the referees for their thorough proofreading and helpful suggestions.

\bibliographystyle{eptcs}
\bibliography{../ACP/acp}

\end{document}